\documentclass[a4paper, amsfonts, amssymb, amsmath, reprint, showkeys, twoside,superscriptaddress]{revtex4-2}
\usepackage[english]{babel}
\usepackage[utf8]{inputenc}
\usepackage[colorinlistoftodos, color=green!40, prependcaption]{todonotes}
\usepackage{subcaption}
\usepackage{ragged2e}

\usepackage[pdftex, pdftitle={Article}, pdfauthor={Author}, pdfborder={0 0 0}]{hyperref} \hypersetup{colorlinks=true, linkcolor=red, citecolor=green, urlcolor=black}
\usepackage{physics}
\usepackage{bbm}
\usepackage{zref-clever}
\usepackage{amsthm}

\newtheorem{lemma}{Lemma}

\zcsetup{font=\color{blue}}
\zcsetup{cap}
\let\oldzcref\zcref
\renewcommand{\zcref}[2][]{%
  \begingroup
    \hypersetup{linkcolor=blue}
    \oldzcref[#1]{#2}%
  \endgroup
}

\newcommand{\mc}[1]{\mathcal{#1}}
\newcommand{\id}{\mathbbm{1}}

\newcommand{\pos}{\mathrm{pos}}
\newcommand{\sgn}{\mathrm{sgn}}
\newcommand{\ee}{\mathrm{e}}
\newcommand{\ii}{\mathrm{i}}
\newcommand{\zz}{\mathbf{z}}
\newcommand{\ZZ}{\mathbf{Z}}
\renewcommand{\dd}{\mathrm{d}}

\newcommand{\PB}[2]{\{#1, #2\}_P}
\newcommand{\MB}[2]{\{#1, #2\}_M}
\renewcommand{\tr}{\operatorname{tr}}

\begin{document}
\title{Tsirelson's nonclassicality witness under dissipative dynamics}
\author{Nguyen Vu Khoi Huynh}\email[Correspondence email address: ]{nguyen.huynh@u.nus.edu}
\affiliation{Centre for Quantum Technologies, National University of Singapore, 3 Science Drive 2, Singapore 117543}
    
\author{Nicky Nel Narido Labayna}
\affiliation{Centre for Quantum Technologies, National University of Singapore, 3 Science Drive 2, Singapore 117543}

\author{Martine Schut}
\affiliation{Centre for Quantum Technologies, National University of Singapore, 3 Science Drive 2, Singapore 117543}

\author{Valerio Scarani}
\affiliation{Centre for Quantum Technologies, National University of Singapore, 3 Science Drive 2, Singapore 117543}
\affiliation{Department of Physics, National University of Singapore, 2 Science Drive 3, Singapore 117542}
\date{\today}

\begin{abstract}
A dynamics-based test, originally proposed by Tsirelson for the harmonic oscillator, provides a method for certifying quantumness under the assumption of a known Hamiltonian. These tests, however, are typically proposed for isolated systems, an assumption that breaks down in experimental implementation. In this paper, we extend the protocol to the harmonic oscillator coupled to standard models of dissipation: thermal relaxation, pure dephasing, and the Caldeira--Leggett model. Using the Moyal-Wigner formalism of quantum mechanics in phase space, we show that the introduction of dissipation requires a dissipation-dependent shift of the classical bound, and provide the threshold under which the nonclassicality witness retains its validity.
\end{abstract}

\maketitle

\section{Introduction}
Certifying nonclassicality in continuous-variable (CV) systems is a fundamental task in quantum science. An indicator of non-Gaussian quantumness is the presence of negative values in the Wigner quasiprobability distribution, which is both a necessary resource for quantum advantage in CV quantum computation~\cite{mari2012} and a central object in the resource theories of non-Gaussianity~\cite{albarelli2018}.
While full tomography can be used to reconstruct the Wigner function, it demands significant resources. 
Certifying Wigner negativity with little resources and simple measurements is thus a relevant and active area \cite{chabaud2021, zaw2024}.
One approach is the dynamics-based protocol introduced by Tsirelson~\cite{tsirelson2006} to assess the non-classical nature of harmonic oscillator, and later developed in \cite{zaw2022} for general uniformly precessing observables, by examining the positivity of its coordinate during its evolution in time and assigning to it a score. 
It is established that any state with a positive Wigner function satisfies a geometric bound and any violation certifies Wigner negativity. 
The appeal of this protocol is that only a single binary measurement per round is required. 
This technique has been adapted beyond simple harmonic motion to other time-independent Hamiltonians~\cite{zaw2023, zaw2022}, entanglement witnessing \cite{jayachandran2023, huynh-vu2024},  and was recently demonstrated in a nuclear spin qudit \cite{vaartjes2025}. It also attracts related foundational questions \cite{garg2026, plavala2024, zaw2025}.
These protocols, however, are defined for ideal, isolated systems.
In any experimental setting, the system will couple to an environment, leading to dissipation and decoherence. 
This environmental interaction leads to the emergence of classicality and smears the Wigner function, degrading the negativity that enables the protocol while also affecting the classical dynamics it benchmarks against. 
This raises two questions: (i) What is the valid classical bound under noisy dynamics? (ii) Up to what noise level can quantum states still violate it?

We address both questions for the harmonic oscillator under three channels spanning distinct decoherence mechanisms in CV quantum hardware: phase noise, energy exchange (or amplitude damping), and force-induced diffusion.
Pure dephasing, thermal relaxation (photon gain/loss), and quantum Brownian motion are the cleanest representatives of these three mechanisms, respectively, and are analytically tractable.
Many realistic environments reduce to one of these in an appropriate limit, photon gain/loss give the heating/cooling rate due to, for example, electric-field noise in ion traps~\cite{brownnutt_ion-trap_2015} or a thermal bath in optomechanics~\cite{Nunnenkamp2011optomechanics}.
Phase noise is an energy-conserving frequency noise which can originate, for example, from voltage noise in surface ion traps~\cite{brownnutt_ion-trap_2015} or charge noise in circuit QED~\cite{haroche2006exploring,LarsoncQED}.
Quantum Brownian motion~\cite{GrabertQBM1983} couples linearly to position and can be caused by stochastic fluctuations in linear forces, such as from collisions with air molecules~\cite{schlosshauer2019quantum} or electric field noise in surface ion traps~\cite{brownnutt_ion-trap_2015}.

In this work, we consider these three noise sources and give, for each of them, (1) an updated classical bound and (2) the noise range over which a quantum violation remains achievable. 
We consider two notions of classicality: that of a proper phase-space distribution, and that of a Wigner positive state.
We will show that for both these notions of classicality, the classical bound changes as a function of the noise strength, showing that knowing the expected noise is necessary for these types of dynamics.
Both the protocol and the noise sources are kept generic: the noise strength is taken as a free parameter, and we consider the case of a harmonic oscillator with three equally spaced measurement times for building the nonclassicality witness. In fact, both the number of measurements and the spacing can be varied~\cite{zaw2023}.
The key point of this work is the integration of generic noises in the dynamical protocol and the resulting change of the nonclassical bounds; for application to a specific experimental setup, one can specialise the protocol and define the relevant noise sources accordingly.
A brief introduction of the precession protocol~\cite{tsirelson2006} is given in ~\zcref{sec:protocol}.
~\zcref[S]{sec:photon-fluctuation,sec:fpe_cl,sec:fpe_deph}, then consider case-by-case the three different dephasing mechanisms (thermal relaxation, quantum Brownian motion and pure dephasing, respectively) and find for each case the classical bound and maximum quantum score. 
Comments on the methodology and the certifiable region/robustness of the optimal state and accumulated results are discussed in~\zcref{sec:results}.

\section{General precession protocol}\label{sec:protocol}
Consider a harmonic oscillator of frequency $\omega_0$. The protocol proceeds as follows~\cite{zaw2022}
\begin{enumerate}
    \item Prepare the same initial state of the oscillator independently in every round.
    \item In each round, sample $k\in \{0,1,2\}$ uniformly, allow the oscillator to evolve under its dynamics until $t_k = kT_0/3$ ($T_0 = 2\pi/\omega_0$) and measure the sign of its position quadrature $q$. Assign outcome $1$ if $q>0$ and $0$ if $q\leq 0$. 
    \item Repeat over many rounds. The protocol score is the average score over multiple rounds.
\end{enumerate}

In what follows, we derive the classical and quantum scores in their respective formalisms and bound their maximum over admissible states. The explicit values for each noise model are computed in~\zcref{sec:photon-fluctuation,sec:fpe_cl,sec:fpe_deph}, where the quantum and classical score functional (denoted $\mathcal{S}$) are maximized to find a bound (denoted $\mathbf{P}$) and compared to the protocol score.

\subsection{Classical score \& dissipative dynamics}
In the original protocol stated for closed systems, the dynamical assumption is a statement about observables---uniform precession of the measured quadrature \cite{zaw2022, zaw2025c}---but under noise, the sign statistics depend on the full evolving distribution, including the diffusion induced by the environment and not only on the precession of the mean alone. We establish the correspondence by fixing the assumption at the level of the phase space density, where the object common to the two theories is the equation of motion in phase space \cite{hernandez2024}. A classical description of the protocol is thus defined by two assumptions
\begin{itemize}
    \item (A1) The oscillator's state is described by a non-negative phase space distribution $f(\zz, 0)$ according to classical mechanics. 
    \item (A2) The dynamics of the system is known. Between preparation and measurement, the state $\zz(t)$ is  a realisation of the classical stochastic process defined by the assumed dynamics of the oscillator: uniform precession at $\omega_0$ when isolated \cite{zaw2022, tsirelson2006} and, for the open channels studied here,  the phase space diffusion with the drift and diffusion of the noise channel, under which the density $f(\zz, t)$ obeys the corresponding Fokker–Planck equation (FPE) (\zcref{app:greens}). 
\end{itemize}

Under these two assumptions, the classical prediction for the protocol score is a linear functional of the initial distribution 
\begin{align}\label{eq:phase_space_score_cl}
    \mc{S}[f] = \frac{1}{3}\sum_{k=0}^{2}\int f(\zz,t_k)\,\Theta(q)\,\dd \zz
\end{align}
where $\Theta$ is the Heaviside step function. 

Given a quantum mechanical model, it remains to find the corresponding classical dynamics assigned by (A2). In the noiseless protocol, the dynamical assumption is common between the two theories: a classical oscillator precesses uniformly, and the quadrature operators in the Heisenberg picture obey the same linear equations of motion as their classical counterpart.
Dissipation preserves this agreement for the channels considered here. In each channel, the Moyal expansion of the Wigner-Weyl image of the equation of motion terminates at second order in phase space derivatives, so the Wigner function obeys a Fokker-Planck equation exactly, with no omitted higher-order $\hbar$ terms (\zcref{app:greens}). 
For general Lindbladians, the quantum-classical correspondence is only approximate, with errors controlled by the diffusion strength \cite{hernandez2024}.
That Fokker-Planck equation is at the same time the standard classical stochastic description of the corresponding environment: damped precession with thermal diffusion, Klein-Kramers dynamics of classical Brownian motion, and action-conserving phase diffusion. 
Assumption (A2) adopts this common dynamics. 
The two theories then share a common dynamics in phase space, as they share uniform precession in the noiseless protocol, and differ only in the states they admit: non-negative densities under (A1) in classical theory, against Wigner quasiprobability distribution, which may take negative values, in quantum theory. 
A measured score above the classical bound, therefore, falsifies the conjunction of (A1) and (A2). 

\subsection{Classical Dirac-delta bound}
Both the unitary and dissipative dynamics considered in this paper are linear in $f$ and preserve non-negativity, so the evolution can be written through a non-negative Green's function $G(\mathbf{z},t\mid \mathbf{z}_0)$:
\begin{equation}\label{eq:green_general}
    f(\mathbf{z},t) = \int G(\mathbf{z}, t \mid \mathbf{z}_0)\, f_0(\mathbf{z}_0)\, \dd\mathbf{z}_0 .
\end{equation}
Substituting~\eqref{eq:green_general} into~\eqref{eq:phase_space_score_cl} expresses the score as the expectation of the score function $s(\zz_0)$,
\begin{equation}\label{eq:P3_convex}
    \mc{S}[f] = \int f(\mathbf{z}_0)\, S(\zz_0)\, \dd\mathbf{z}_0,
\end{equation}
where
\begin{equation}\label{eq:P3_dirac}
    S(\zz_0) :=\mc{S}[\delta_{\mathbf{z}_0}] = \frac{1}{3}\sum_{k=0}^{2}\int \Theta(q)\, G(\mathbf{z}, t_k \mid \mathbf{z}_0)\, \dd\mathbf{z}
\end{equation}
is the score of a Dirac initial condition at $\mathbf{z}_0$. 

Since $\mc{S}$ is linear in the initial distribution, its supremum over normalised non-negative distributions is attained at an extreme point: a Dirac delta,
\begin{equation}\label{eq:Pcl_max}
    \mathbf{P}_\text{cl} = \sup_{\mathbf{z}_0 \in \mathbb{R}^2}\,  S(\zz_0)
\end{equation}
Finding the classical bound on the protocol score, $\mathbf{P}_\text{cl}$, reduces to the optimization over Dirac initial conditions, which we evaluate explicitly for each channel using the corresponding Green's function. Any classical state then has a score at most the bound $\mathbf{P}_\text{cl}$ defined in~\eqref{eq:Pcl_max}.

The bound $\mathbf{P}_\text{cl}$ puts the least constrictions on the meaning of \textit{classicality}. The broadest classical description assigns to the oscillator a non-negative phase space density $f\ge 0$ evolving under the assumed dissipative dynamics, with no further restriction. The largest score this classical theory permits is the Dirac-delta bound of~\eqref{eq:Pcl_max}, obtained from a point-particle evolution. This bound is dictated by classical statistical mechanics alone, so a measured $\mc{S} > \mathbf{P}_\text{cl}(\gamma)$ cannot be reproduced by any classical description of the dynamics and certifies nonclassicality; it is the strictest reading of the protocol.

\subsection{Wigner positivity bound}
As much as they provide the most stringent criterion, Dirac deltas in phase space are unphysical, as they violate the uncertainty relations. Different thresholds than the Dirac-delta bound can be introduced to certify other properties of the state, such as non-Gaussianity\footnotemark[1]\label{footnote:non_gaussian}\footnotetext[1]{The threshold for detecting non-Gaussianity, $\mathbf{P}_{\rm ng}$, is not given in this work but can be found by optimizing the score over all Gaussian states, i.e., mixture of displaced squeezed vacuum. By convexity, it is achieved at a pure state, which are fully characterized by $2$ complex parameters. The optimization problem can be set up similarly to the Dirac case by exploiting the Gaussianity of the dynamics in the thermal models (see, e.g.,~\eqref{eq:Pcl_photon}), which reduces to a search over the $4$-dimensional parameter space. Every Gaussian state is Wigner-positive, so its score cannot exceed the Wigner-positive supremum; this threshold therefore fits into the hierarchy $\mathbf{P}_{\rm ng} \leq \mathbf{P}_{\rm wp} \leq  \mathbf{P}_{\rm cl}$.\label{footnote:non_gaussian}} and Wigner negativity.

Wigner negativity is specifically a quantum notion. The Wigner function
$W_\rho$ is a complete description of the state with the correct quadrature marginals and inherits constraints from the underlying density operator (e.g., Heisenberg uncertainty). Its negativity is a signature of nonclassicality inaccessible to any classical phase space distribution $f \geq 0$ and is the nonclassical feature of interest here. The sharp threshold for certification of Wigner negativity is the exact Wigner-positive supremum:
\begin{align}\label{eq:def_P+}
    \mathbf{P}_{\rm wp, e} := \sup_{W_\rho \geq 0} \mc{S}[W_\rho]\,,
\end{align}
the largest score attainable by any physical state with a positive Wigner function. 

Any state with $W_\rho \geq 0$ everywhere is simultaneously a valid classical distribution, so its score cannot exceed $\mathbf{P}_\text{cl}$ and $\mathbf{P}_{\rm wp, e}\leq \mathbf{P}_{\rm cl}$. Therefore, an observed score $\mc{S} > \mathbf{P}_\text{cl}$ certifies Wigner negativity. 
For the noiseless harmonic oscillator, a geometric argument gives $\mathbf{P}_\text{cl} = 2/3$~\cite{tsirelson2006, zaw2022}. This bound is also asymptotically achievable by a coherent state with arbitrarily large amplitude, so the gap between the two bounds is closed: $\mathbf{P}_{\rm cl} = \mathbf{P}_{\rm wp,e} = 2/3$ in the noiseless case. 
Under dissipative dynamics, their equality is no longer guaranteed and each witnesses a different notion of nonclassicality: $\mathbf{P}_{\rm cl}$ is the threshold an experiment must beat to rule out classical theory outright, $\mathbf{P}_{\rm wp,e}$ is the threshold it must beat to witness negativity. We characterize $\mathbf{P}_{\rm cl}$ and provide an upper bound for $\mathbf{P}_{\rm wp, e}$ for each dissipation model in the following sections.

The bound $\mathbf{P}_{\rm wp,e}$ corresponds to the tightest witness for Wigner negativity. However, computing it is hard: it is a linear functional maximized over the set of Wigner-positive states, which is not well characterized \cite{vanherstraeten2025}.  A similar optimization problem to \eqref{eq:def_P+} was considered in \cite{chabaud2021} through SDP relaxation for diagonal observables in the Fock basis, a condition that fails for our score observable. We instead consider tractable relaxations obtained by discarding constraints on the Wigner function and maximizing the score over the enlarged feasible set. An extreme choice is to discard all quantum constraints and keep only the positivity and normalization conditions, yielding the set of all classical phase space distributions and the bound $\mathbf{P}_\text{cl}$ considered above. This bound, however, is achieved at a Dirac distribution, which is maximally localized and highly unphysical. This leads to an unnecessarily strict bound for witnessing Wigner negativity. Keeping an additional uniform cap: 
\begin{align}\label{eq:wigner_cap}
    |W(q,p)|\leq \frac{1}{\pi \hbar}\,,
\end{align}
that every Wigner function obeys~\cite{curtright2013} tightens it considerably:
\begin{align}\label{eq:bathtub bound}
    \mathbf{P}_{\rm wp,e} \leq \mathbf{P}_{\rm wp} := \max_{0\leq f \leq \frac{1}{\pi \hbar}}\mc{S}[f] \leq \mathbf{P}_{\rm cl} 
\end{align} 
From \eqref{eq:P3_convex}, the score functional is linear in the initial distribution. Therefore, maximizing it over the probability distributions respecting the cap in \eqref{eq:bathtub bound} can be done via the bathtub principle \cite{lieb2001}: the optimal distribution saturates the cap on the highest-scoring region and zero elsewhere,
\begin{align}\label{eq:bathtub_optimizer}
    f^\star := \frac{1}{\pi \hbar}\bqty{\mathbbm{1}_{\{S>\lambda^\star\}} + c\mathbbm{1}_{\{S=\lambda^\star\}}} \,.
\end{align}
where $\lambda^\star$ is the level satisfying $|\{S>\lambda^\star\}|\leq \pi \hbar \leq |\{S\geq\lambda^\star\}|$ for $S$ in~\eqref{eq:P3_dirac}, and $c\in [0,1]$ is chosen to ensure normalization. For the two thermal models treated below, the score has no plateau, i.e., the level set $\{s=\lambda^\star\}$ has measure zero, and the second term in \eqref{eq:bathtub_optimizer} vanishes. \newline

\begin{figure*}[htp!]
    \centering\includegraphics[width=\linewidth]{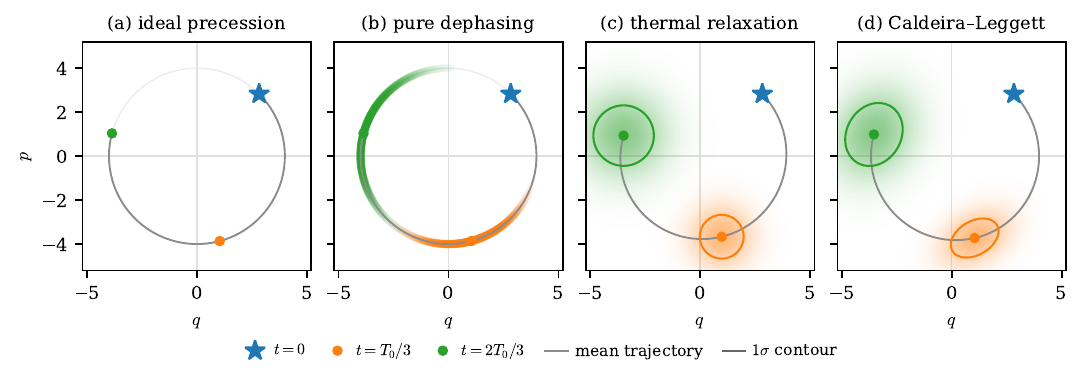}
    \caption{\justifying Phase space evolution of an initial Dirac. The point/star indicates the three equally spaced measurement points. 
    The figure gives an intuitive understanding of how, for each of the three noise models, the dissipation affects the probability of finding the particle in $\{x>0\}$, showing that the classical score can change under dissipative dynamics.
    For pure dephasing (\zcref{sec:fpe_deph}) there is angular diffusion, for thermal relaxation (\zcref{sec:photon-fluctuation}) and in the Caldeira-Leggett model (\zcref{sec:fpe_cl}) the trajectory in-spirals and the state's variance increases until it thermalizes.}
    \label{fig:phase_space_evolution}
\end{figure*}

The two computable upper bounds $\mathbf{P}_\text{cl}$ and $\mathbf{P}_\text{wp}$ behave qualitatively differently at weak noise. 
The Dirac optimizer is the infinitely localized point mass, whose score is fixed by a variance ratio alone and carries no absolute phase space scale; this makes $\mathbf{P}_{\rm cl}(\gamma)$ jump discontinuously above $2/3$ as $\gamma\to 0^+$ in both thermal channels. The cap \eqref{eq:wigner_cap} excludes the point mass and restores the missing scale, allowing $\mathbf{P}_{\rm wp}$ to vary continuously from the noiseless bound $2/3$ in the same limit. Pure dephasing behaves differently from both thermal channels. Its score function is independent of the radius, so a coherent state at the optimal angle achieves a score arbitrarily close to the classical bound in the limit of large radius (\zcref{sec:dephasing_classical}), and the two bounds coincide.

\subsection{Quantum score and bound}\label{sec:quantum_bound}
In the standard formalism, a state is a density operator $\rho$ and the score reads
\begin{equation}\label{eq:Pqm_op}
    \mc{S}[\rho] = \frac{1}{3}\sum_{k=0}^{2}\tr\bqty{\rho(t_k)\,\pos(\hat q)} = \tr\bqty{\rho\, \hat S},
\end{equation}
where the score operator
\begin{equation}\label{eq:Q_def}
    \hat S = \frac{1}{3}\sum_{k=0}^{2} \pos(\hat q; t_k),
    \qquad
    \pos(\hat q; t) := \ee^{\mc{L}^\dag t}\bqty{\pos(\hat q)},
\end{equation}
is the equally-weighted average of the Heisenberg-evolved positivity projector $\pos(\hat q) = \frac{1}{2}(\id + \sgn(\hat q))$ under the adjoint evolution generator $\mc{L}^\dag$. It is the quantisation of $S$ in~\eqref{eq:P3_dirac}. The maximum quantum score is the largest eigenvalue of $\hat S$,
\begin{equation}\label{eq:Pqm_max}
    \mathbf{P}_\text{qm} = \max_\rho \tr(\rho \hat S) = \lambda_{\max}(\hat S),
\end{equation}
attained by the corresponding eigenstate. 
In~\zcref{sec:photon-fluctuation,sec:fpe_deph,sec:fpe_cl} $\pos(\hat q;t)$ is determined for each noise model (i.e. for specific $\mathcal{L}^\dagger$), after which $\hat S$ is diagonalised numerically to find the maximum possible quantum score.

It is illustrative to rewrite the quantum score in the Moyal--Wigner formalism:
\begin{align}\label{eq:phase_space_score_qm}
    \mc{S}[\rho] &= \frac{1}{3}\sum_{k=0}^{2}\int W_\rho(q,p,t_k)\, \Theta(q)\, \dd q\, \dd p = \mc{S}[W_\rho]
\end{align}
In this formalism, the quantum score takes the same form as~\eqref{eq:phase_space_score_cl} with $f \to W_\rho$.
While $f$ is a non-negative phase space distribution, $W_\rho$ is a quasi-probability which can contain negative values, allowing for a violation of the classical score. 

The closed-system threshold $2/3$ assumes only uniform precession and no further model of the dynamics enters. Under noise, the dynamical assumption (A2), by contrast, is an assumption of both the underlying dissipative model and its parameters. An illustration of the phase space evolution under the different noise models is given in~\zcref{fig:phase_space_evolution}. Within the weak-noise windows studied below, the classical thresholds increase monotonically with the noise scale, so a conservative experiment may calibrate an upper bound on the noise and get an upper bound for the certification threshold. A violation then excludes every classical distribution evolving under any noise strength within the calibrated range. 

The three quantities introduced above---the maximum quantum score $\mathbf{P}_\text{qm}(\gamma)$ and the two computable bounds $\mathbf{P}_\text{cl}(\gamma)$ and $\mathbf{P}_\text{wp}(\gamma)$---are all functions of the noise strength $\gamma$ in the respective model. We call a noise strength certifiable when there exists at least one quantum state that violates the bound. In that case, an optimally prepared state achieves a score that no classical distribution (for $\mathbf{P}_{\rm cl}$) or Wigner-positive state (for $\mathbf{P}_{\rm wp}$) can reproduce. The set of all certifiable $\gamma$ is the certifiable region. In the regime we consider, where the maximum quantum score $\mathbf{P}_\text{qm}$ decays and the bounds rise with noise, the region is an interval $[0,\gamma^\star)$ whose upper edge is the crossing rate $\gamma^\star$, defined by 
\begin{align}\label{eq:gamma_star}
    \mathbf{P}_\text{qm}(\gamma^\star) = \mathbf{P}(\gamma^\star);
\end{align}
Each channel therefore carries two crossing rates, $\gamma^\star(\mathbf{P}_\text{cl})$ and $\gamma^\star(\mathbf{P}_\text{wp})$, bounding the certifiable regions for nonclassicality and for Wigner negativity respectively; because $\mathbf{P}_\text{wp} \le \mathbf{P}_\text{cl}$, the latter region always contains the former. We report the certifiable regions for each thermal model in their respective sections. In the pure dephasing case, we will later show in \zcref{sec:fpe_deph} that there is always a quantum state that achieves $\mathbf{P}_{\rm cl}$ at least asymptotically.

\section{Case 1: Thermal relaxation}\label{sec:photon-fluctuation}

Amplitude damping by a thermal bath at mean thermal occupation number $\bar n$ is the standard model of energy exchange in optical cavities, motional modes, and microwave cavities~\cite{breuer2002}. 
Under the Born--Markov approximation and rotating-wave approximation, the reduced dynamics are described by the Lindblad master equation~\cite{ carmichael1998}
\begin{equation}\label{eq:lindblad_photon}
    \dot{\hat\rho} = \mc{L}[\hat{\rho}] :=  -\frac{\ii}{\hbar}[\hat H, \hat\rho]
        + \gamma_0(\bar n + 1)\,\mathcal{D}[\hat a]\hat\rho
        + \gamma_0\bar n\,\mathcal{D}[\hat a^\dag]\hat\rho,
\end{equation}
with $\hat H = \hbar\omega \hat a^\dagger \hat a, ~\omega:=(1+\delta)\omega_0$~\footnote[2]{Under thermal relaxation, the frequency shift is often negligible in practice~\cite{carmichael1998}.} is the effective Hamiltonian after taking into account the Lamb shift $\delta \omega_0$. The following analysis assumes zero Lamb shift, see~\eqref{eq:lindblad_photon}; robustness under misspecified frequency shift is discussed in~\zcref{app:freq_shift}. The dissipator $\mathcal{D}[\hat L]\hat\rho := \hat L\hat\rho\hat L^\dag - \tfrac{1}{2}\{\hat L^\dag\hat L, \hat\rho\}$ has Lindblad jump operators $\hat{L}_- = \sqrt{\gamma_0(\bar{n}+1)}\hat{a}$ and $\hat{L}_+ = \sqrt{\gamma_0\bar{n}}\hat{a}^\dag$. 
Since both operators have linear Weyl symbols (i.e., their Wigner--Weyl corresponding phase space functions), the Wigner--Weyl mapping of the dissipator terminates exactly at second order in derivatives~\cite{dubois2021,hernandez2024}. The master equation~\eqref{eq:lindblad_photon} therefore reduces to a Fokker--Planck equation of an Ornstein–Uhlenbeck (OU) process~\cite{risken2012, carmichael1998} (see also~\zcref{app:greens}):
\begin{equation}\label{eq:fpe}
\frac{\partial W}{\partial t} = -\nabla\cdot(\mathbf{A}\mathbf{z}\,W) + \nabla \cdot \pqty{\boldsymbol{D} \nabla W} ,
\end{equation}
with drift matrix
\begin{equation}\label{eq:drift_photon}
\mathbf{A} = \begin{pmatrix} - \frac{\gamma_0}{2} & 1/m   \\[3pt] -m\omega_0^2 & - \frac{\gamma_0}{2} \end{pmatrix},
\end{equation}
and diffusion matrix
\begin{equation}\label{eq:D0}
\boldsymbol{D} = \frac{\hbar\gamma_0(2\bar{n}+1)}{4} \begin{pmatrix}  1/m\omega_0 &0  \\[3pt] 0 &  m\omega_0 \end{pmatrix} \, .
\end{equation}
In the semiclassical limit $\hbar \to 0$, the mean photon number scales as $\bar{n} \approx k_B T/\hbar \omega_0$, dominating the vacuum fluctuation. The diffusion term then becomes
\begin{align*}
    \boldsymbol{D} \approx \frac{\gamma_0}{2\omega_0} k_B T\begin{pmatrix}  1/m\omega_0 &0  \\[3pt] 0 &  m\omega_0 \end{pmatrix}  \,.
\end{align*}
The diffusion does not vanish in this limit but converges to classical diffusion.

Since the drift and diffusion are constant, the dynamics is Gaussian-preserving and the corresponding stochastic trajectory has a Green's function that is a Gaussian:
\begin{equation}\label{eq:green_photon}
    G(\mathbf z, t \mid \mathbf z_0) = \frac{\exp\bqty{-\tfrac{1}{2}\pqty{\mathbf z - \bar{\mathbf z}(t)}^T \mathbf C(t)^{-1}\pqty{\mathbf z - \bar{\mathbf z}(t)}}}{2\pi\sqrt{\det\mathbf C(t)}},
\end{equation}
with spiraling center
\begin{equation}\label{eq:spiral_centre}
    \bar{\mathbf z}(t) = \ee^{-\gamma_0 t/2}\,R(\omega_0 t)\,\mathbf z_0 =: (\bar q(t), \bar p(t))
\end{equation}
where
\begin{align*}
    \qquad R(\phi) = \begin{pmatrix}\cos(\phi) & \sin(\phi)/(m\omega_0)\\ -m\omega_0\sin(\phi) & \cos(\phi)\end{pmatrix},
\end{align*}
and variance
\begin{equation}\label{eq:variance_photon}
    \mathbf C(t) = \frac{\hbar(2\bar n + 1)}{2}\bigl(1 - \ee^{-\gamma_0 t}\bigr)\,
    \begin{pmatrix} 1/(m\omega_0) & 0 \\[3pt] 0 & m\omega_0 \end{pmatrix}.
\end{equation}
The detailed derivation is given in~\zcref{app:greens}.
The variance grows from zero towards the thermal value, showing thermalization of the state under thermal relaxations.
The centre of the Gaussian shows a deterministic spiral, which is also shown in the phase space evolution in~\zcref{fig:phase_space_evolution}.

 \subsection{Classical Bounds}\label{subsec:thermal_classical}
Substituting the Green's function of~\eqref{eq:green_photon} into the Dirac trajectory score $\mc{S}$, given in~\eqref{eq:P3_dirac} and integrating out $p$ leaves a Gaussian marginal in $q$ whose half-line probability is the standard normal cumulative distribution function $\Phi$:
\begin{equation}\label{eq:P3cl_photon}
    S(\zz_0) = \frac{1}{3}\sum_{k=0}^{2}\Phi\!\pqty{\frac{\bar{q}(t_k)}{\sigma_q(t_k)}}.
\end{equation}
where we denote the variance of position marginal by 
\begin{align}
    \bqty{\sigma_q}^2(t) = \dfrac{\hbar(2\bar n + 1)}{2 m\omega_0}\bigl(1 - \ee^{-\gamma_0 t}\bigr) \label{eq:sigma_q}
\end{align}
which is the matrix element $C_{qq}$ of the covariance matrix $\mathbf{C}$ in \zcref{eq:variance_photon}.
At the first probing time $t=0$, the variance vanishes, and the quantile collapses to the Heaviside step function $\Theta(q_0)$. 

The classical bound follows from a two-dimensional optimization over the initial point $\mathbf{z}_0$:
\begin{equation}\label{eq:Pcl_photon}
    \mathbf{P}_\text{cl}(\gamma_0, \bar{n}) = \sup_{\mathbf{z}_0 \in \mathbb{R}^2}\frac{1}{3}\sum_{k=0}^{2}\Phi\!\pqty{\frac{\bar{q}(t_k)}{\sigma_q(t_k)}} = \mathbf{P}_\text{cl}(\gamma_0)
\end{equation}
independent of the thermal parameter $\bar n$ or $T$ (\zcref{fig:bounds_thermal}(b)) because the objective function can be made independent of them by a change of variable $\zz_0 \mapsto \zz_0/\sqrt{\bar n +\frac{1}{2}}$. This optimization admits a closed form(\zcref{app:limits}).

In the noiseless limit $\gamma_0 \to 0$, we have $\sigma_q(t) \to 0$ (see~\eqref{eq:sigma_q}), and as a result~\eqref{eq:P3cl_photon} reduces to the average of indicator functions of three half-planes $\frac{1}{3}\sum_k \Theta(\bar{q}(t_k))$, with at most two indicators simultaneously $+1$. The maximum then recovers Tsirelson's geometric bound $\mathbf{P}_\text{cl} = 2/3$ \cite{tsirelson2006}. 

For finite noise, the diffusion smears each trajectory, and a Dirac concentrated near the common boundary of the three half-planes acquires a nonzero probability of being measured positive at each of the three probing times. Interestingly, even an infinitesimal positive dissipation lifts the bound above the noiseless value $2/3$, with no dependence on bath temperature:
\begin{equation}\label{eq:Pcl_jump_photon}
    \mathbf{P}_\text{cl} (\gamma_0\to 0^+, \bar{n}) = \frac{1}{3}\bqty{1 + \Phi\pqty{\sqrt{2\ln 2}} + \Phi\pqty{-\sqrt{\ln 2}}},
\end{equation}
which evaluates to $\approx 0.6943$ (\zcref{app:limits}). The noiseless geometric bound $2/3$ is not robust: any arbitrarily small coupling to a thermal bath lifts the supremum over non-negative phase space distributions above $2/3$ by a margin of $\approx 0.028$. 

In the opposite limit $\gamma_0 \to \infty$, two of the three marginals collapse onto $\{q > 0\}$, giving the score of $1$, while the third thermalizes to a symmetric configuration, giving the score of $1/2$ (\zcref{app:limits}). Together, the score in high noise limit is
\begin{equation}
    \mathbf{P}_{\rm cl} (\gamma_0 \to \infty, \bar{n}) = \frac{1}{3}\pqty{1+1+\frac{1}{2}} = \frac{5}{6} \,.
\end{equation}
This asymptote lies outside the weak-coupling regime $\gamma_0/\omega_0 \ll 1$ in which the Lindblad equation is derived, and is not used in the certification plots. Detailed derivation of these results is given in~\zcref{app:limits}.

The discontinuous jump for infinitesimal noise is due to the physical scale blindness of the Dirac optimizer: the score in the limit $\gamma \to 0^+$ depends only on the variance ratio $\sigma_q(t_2)/\sigma_q(t_1) \to \sqrt{2}$ (\zcref{app:limits}) and carries no phase space scale. 
The configuration that realizes it is a point mass of infinite phase space density and zero area, which is an unphysical state in quantum theory. Imposing the density cap $W\leq 1/(\pi \hbar)$~\eqref{eq:wigner_cap} brings back the physical scale and allows for a tighter bound. The bound is found by substituting the optimizer \eqref{eq:bathtub_optimizer} at the optimal level set $\lambda^\star$ into the score functional in \eqref{eq:phase_space_score_cl}, which then reduces to a one-dimensional integral that is evaluated numerically. At weak coupling,  $\lambda^\star - 2/3$ is exponentially small in $[(\bar n + 1/2)\gamma_0/\omega_0]^{-1}$ (\zcref{app:wp_thermal}) and
\begin{align}\label{eq:photon_bathtub_slope}
    \mathbf{P}_{\rm wp}(\gamma_0\to 0^+, \bar n) = \frac{2}{3} + \frac{3-2\sqrt{2}}{9\sqrt{3}}\pqty{\bar n +\frac{1}{2}}\frac{\gamma_0}{\omega_0} + \mc{O}(\gamma_0^2)\,.
\end{align}
In contrast to the scale-blind Dirac bound, the Wigner positivity bound rises continuously at $\gamma_0 = 0$ with a slope proportional to $(\bar n + 1/2)$.
\zcref[S]{fig:bounds_thermal}(a) plots the Wigner positivity bound and shows that its slope increases for increasing occupation number $\bar{n}$.

\begin{figure*}
    \centering
    \includegraphics[width=\linewidth]{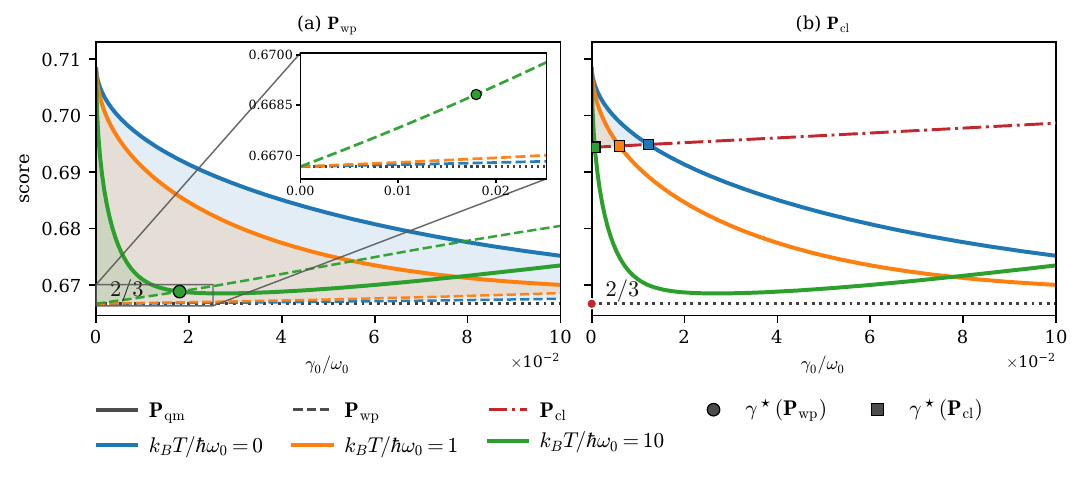}
    \caption{\justifying Score thresholds for the thermal relaxation channel at $k_B T/\hbar\omega_0 = 0, 1, 10$ (colour). It compares the quantum maximum $\mathbf{P}^{\rm th}_\text{qm} = \lambda_{\max}(\hat S)$ (thick solid,~\eqref{eq:Q_photon}) against (a) the Wigner positivity bound $\mathbf{P}^{\rm th}_{\rm wp}$ (dashed,~\eqref{eq:bathtub bound}) and (b) against the Dirac bound $\mathbf{P}^{\rm th}_\text{cl}$ (dash-dotted,~\eqref{eq:Pcl_photon}).}
    \label{fig:bounds_thermal}
\end{figure*}

\subsection{Maximum Quantum Score}
The quantum score is found by analytically finding $\hat S$ defined in~\eqref{eq:Q_def} and then numerically diagonalizing it to find the maximum value.
The score operator $\hat S$ is given by the Heisenberg-evolved $\pos(\hat q;t) = \ee^{\mc{L}^\dag t}[\pos(\hat q)]$. 
Because the dissipator is at most quadratic in canonical operators, the adjoint evolution can be solved in closed form on Weyl symbols. The Weyl symbol of the operator $\pos(\hat q)$ (denoted by the subscript $W$) is the Heaviside step $\Theta(q)$, and its evolution is given by the convolution with the Green's function $G$ of the adjoint dynamics
\begin{align}
    &\pqty{\pos(\hat q;t)}_W(\mathbf{z}) 
    = \int \Theta(q')\, G(\mathbf z', t\mid \mathbf z)\,\dd^2 \mathbf z'. \label{eq:green convol} \\
    &\text{using: } \pos(\hat{q}) = \int_0^{\infty} \ketbra{q'} \dd q'  \leftrightarrow \int_0^{\infty} \delta(q-q') \dd q' = \Theta(q) \nonumber
\end{align}
The Green's function is Gaussian in $\mathbf z'$ with mean 
\begin{equation}\label{eq:qbar_photon}
    \bar{q}(q,p, t) = \ee^{-\gamma_0 t/2}\bqty{q\cos(\omega_0 t) + \frac{p}{m\omega_0}\sin(\omega_0 t)}
\end{equation}
and variance $\sigma_q^2(t)$ of~\eqref{eq:sigma_q}. Therefore,~\eqref{eq:green convol} evaluates to
\begin{equation}\label{eq:posq_symbol_photon}
    \pqty{\pos(\hat q; t)}_W(q,p) = \Phi\!\pqty{\frac{\bar{q}(q,p, t)}{\sigma_q(t)}}.
\end{equation}
with $\Phi$ the standard normal cumulative distribution function.
Since the right-hand side depends on $(q,p)$ only through a linear combination, the Weyl quantization is unambiguous: any function of a linear combination $\hat q$ and $\hat p$ coincides with the Weyl quantisation of the same function applied to its symbol \cite{mccoy1932, hernandez2024}. Defining the rotated, damped quadrature (the subscript H denoting the quantum Heisenberg-evolved operator)
\begin{equation}\label{eq:qH_photon}
    \hat q_H(t) := \ee^{-\gamma_0 t/2}\bqty{\hat q\cos(\omega_0 t) + \frac{\hat p}{m\omega_0} \sin(\omega_0 t)},
\end{equation}
we obtain the closed-form score operator
\begin{equation}\label{eq:posq_op_photon}
    \pos(\hat q; t) = \Phi\!\pqty{\frac{\hat q_H(t)}{\sigma_q(t)}},
\end{equation}
Again, in the closed-system limit $\gamma_0 \to 0$ gives $\sigma_q(t) \to 0$, see~\eqref{eq:sigma_q}, which recovers the noiseless Heisenberg evolution $\pos(\hat q;t) \to \pos(\hat q\cos\omega_0 t + \frac{\hat p}{m\omega_0}\sin\omega_0 t)$.

The score operator is therefore the explicit function of canonical operators
\begin{equation}\label{eq:Q_photon}
    \hat S = \frac{1}{3}\sum_{k=0}^{2} \Phi\!\pqty{\frac{\hat q_H(t_k)}{\sigma_q(t_k)}},
\end{equation}
which we numerically diagonalise in a truncated Fock basis to obtain the maximum quantum score $\mathbf{P}_\text{qm}(\gamma_0, \bar{n}) = \lambda_{\max}(\hat S)$ plotted in~\zcref{fig:bounds_thermal}.

\subsection{Result}
\zcref[S]{fig:bounds_thermal} compares the noise-dependent quantum maximum $\mathbf{P}_\text{qm}(\gamma_0) = \lambda_{\max}(\hat S)$ versus the classical bound $\mathbf{P}_\text{cl}(\gamma_0)$, and the Wigner positivity bound $\mathbf{P}_{\rm wp}$. 
As discussed in~\zcref{subsec:thermal_classical}, the conservative Dirac bound $\mathbf{P}_\text{cl}(\gamma_0)$ jumps discontinuously to $\approx 0.6943$ \eqref{eq:Pcl_jump_photon} and sits just below the noiseless optimum $\mathbf{P}_{\rm qm}(0) \approx 0.709$, so the certifiable region $\mathbf{P}_{\rm qm}> \mathbf{P}_{\rm cl}$ closes rapidly. The certifiable region boundary $\gamma_0^\star/\omega_0$ (\zcref{sec:quantum_bound}), which is the crossing point of the two curves, lies at $1.22\times 10^{-2}$, $5.95\times 10^{-3}$ and $6.72\times 10^{-4}$ for $k_BT/\hbar \omega_0 = 0, 1$ and $10$, respectively.
The scale-aware Wigner positivity bound, by contrast, leaves $2/3$ continuously, so the certifiable region for Wigner negativity using our Wigner negativity bound $\mathbf{P}_{\rm qm}> \mathbf{P}_{\rm wp}$ remains open across the entire rotating-wave regime $\gamma_0/\omega_0 \lessapprox 0.1$ at temperatures $k_B T/\hbar \omega_0 \in \{0,1\}$ and closes only at $\gamma_0^\star (\mathbf{P}_{\rm wp}) \approx 2\times 10^{-2}$ for higher temperature such as $k_B T/\hbar \omega_0 = 10$, which is one to two orders of magnitude beyond the Dirac crossing. The temperature dependence of the certifiable region is deferred to \zcref{sec:temp_dep}. The temperature dependence enters $\mathbf{P}_{\rm wp}$ through the slope while, as mentioned, the Dirac-delta bound is temperature independent due to its scale: one can rescale the optimal Dirac distribution by $\sqrt{\bar n + 1/2}$ to account for the scale change and achieve the same score $\mathbf{P}_{\rm cl}$ (\zcref{eq:Pcl_photon}). 

The condition $\gamma_0/\omega_0\ll1$ essentially requires a high quality-factor, requiring many coherent oscillations before scattering with a photon.
Ion trapping experiments typically operate around $\omega_0\sim2\pi\cross1$~MHz, which means that for our analysis to be applicable, heating rates of $\gamma_0<10^6$~Hz are allowed. This falls within the typical surface-induced heating rates at room temperature but does put constraints on the trap size~\cite{brownnutt_ion-trap_2015}\footnote[3]{Considering the total heating rate to be $\gamma_h = \gamma_0 \bar{n}$, assuming zero frequency shift.}.
The condition $\gamma_0/\omega_0\ll1$ is also generally satisfied by circuit and cavity QED platforms because the effective atom-field coupling and photon decay rate ($\sim$kHz$-$MHz) are typically smaller than the photon and atomic transition frequencies ($\sim$GHz$-$THz)~\cite{Blais2021cqed,haroche2006exploring,LarsoncQED}.

\section{Case 2: Quantum Brownian Motion (Caldeira-Leggett)}\label{sec:fpe_cl}
The Caldeira-Leggett model is a standard noise model describing the system of a harmonic Hamiltonian coupled to an Ohmic bath~\cite{caldeira1983}. In the high-temperature limit ($k_B T \gg \hbar \omega_0$) of this coupling, the dynamics are governed by the Caldeira-Leggett master equation:
\begin{equation}
    \dot{\hat\rho} = -\frac{\ii}{\hbar}[\hat H_0, \hat\rho]
        - \frac{\ii\gamma}{2\hbar}[\hat q, \{\hat p, \hat\rho\}]
        - \frac{m\gamma k_B T}{\hbar^2}[\hat q, [\hat q, \hat\rho]],
    \label{eq:caldeira_leggett}
\end{equation}
with free Hamiltonian $\hat H_0 = \hbar\omega_0 \hat a^\dagger \hat a$ damping rate $\gamma$. Unlike the master equations for the thermal relaxations in~\eqref{eq:lindblad_photon} and the pure dephasing that will be introduced in~\eqref{eq:lindblad_dephasing}, this equation is not of Lindblad form: the high-$T$ approximation breaks complete positivity at sub-thermal timescales~\cite{breuer2002}. 

Mapping the Caldeira-Leggett master equation to the Wigner phase space yields a Fokker--Planck equation of the Ornstein-Uhlenbeck (OU) form in~\eqref{eq:fpe} with drift and diffusion matrix:
\begin{align}\label{eq:drift_cl}
    \mathbf{A}= \begin{pmatrix} 0 & 1/m \\ -m\omega_0^2 & - \gamma  \end{pmatrix}, \qq{}
    \mathbf{D} = \begin{pmatrix} 0 & 0 \\ 0 & m\gamma k_B T \end{pmatrix}
\end{align}
This is the classical Klein-Kramers equation~\cite{risken2012}. Since the dynamics is Gaussian, the derivation of the scores in the Caldeira--Leggett model is the same as for the thermal relaxations discussed before and we refer to~\eqref{eq:green_photon}--\eqref{eq:Pcl_photon}.
Again, the classical bound is given by the supremum of the normal cumulative distribution functions at the measurements times, but with different mean and variance compared~\zcref{sec:photon-fluctuation}. For the underdamped regime $\gamma< 2\omega_0$, define $\Omega = \sqrt{\omega_0^2 - \gamma^2/4}$.
The Green's function for an initial Dirac located at $\mathbf{z}_0 = (q_0, p_0)$ is a Gaussian, now centred on the trajectory
\begin{equation}\label{eq:mean_cl}
    \bar{q}(t) = \ee^{-\gamma t/2}\bqty{q_0\cos(\Omega t) + \frac{p_0/m + \gamma q_0/2}{\Omega}\sin(\Omega t)},
\end{equation}
where the position variance is given in closed form by
\begin{align}\label{eq:var_cl}
    (\sigma_q)^2 =  \frac{k_BT}{m \omega_0^2}\Bqty{1 - \ee^{-\gamma t}\bqty{1 + \frac{\gamma}{2\Omega}\sin(2\Omega t) + \frac{\gamma^2}{2\Omega^2}\sin^2(\Omega t) }}
\end{align}
See appendix~\ref{app:greens} for the derivation of this position variance.
One notable difference with the thermal relaxation is that only the momentum coordinate is diffusively driven~\eqref{eq:drift_cl}, and radial contraction is replaced by underdamped harmonic motion at the effective frequency $\Omega$, see~\zcref{fig:phase_space_evolution}.

\subsection{Classical \& Quantum Scores}
The classical bound is found by substituting~\eqref{eq:mean_cl} into the classical score of~\eqref{eq:P3_dirac} and integrating over $p$, this gives:
\begin{equation}\label{eq:P3cl_cl}
    S(\zz_0) = \frac{1}{3}\sum_{k=0}^{2} \Phi\!\pqty{\frac{\bar q(t_k)}{\sigma_q(t_k)}},
\end{equation}
where $\sigma_q(t_k)$ is the standard deviation of the position marginal in~\eqref{eq:mean_cl} and $\bar q(t_k)$ is the mean position in~\eqref{eq:mean_cl}. The quantum Brownian motion will give the same inspiral-diffusion behaviour as the thermal relaxations, but with different $\bar q(t_k)$ and $\sigma_q(t_k)$ with respect to the thermal relaxation.
In the noiseless limit $\gamma= 0$, assuming the temperature stays fixed, $\sigma_q(t)$ becomes zero and the bound returns to $2/3$.

\begin{figure*}[htp!]
    \centering
    \includegraphics[width=\linewidth]{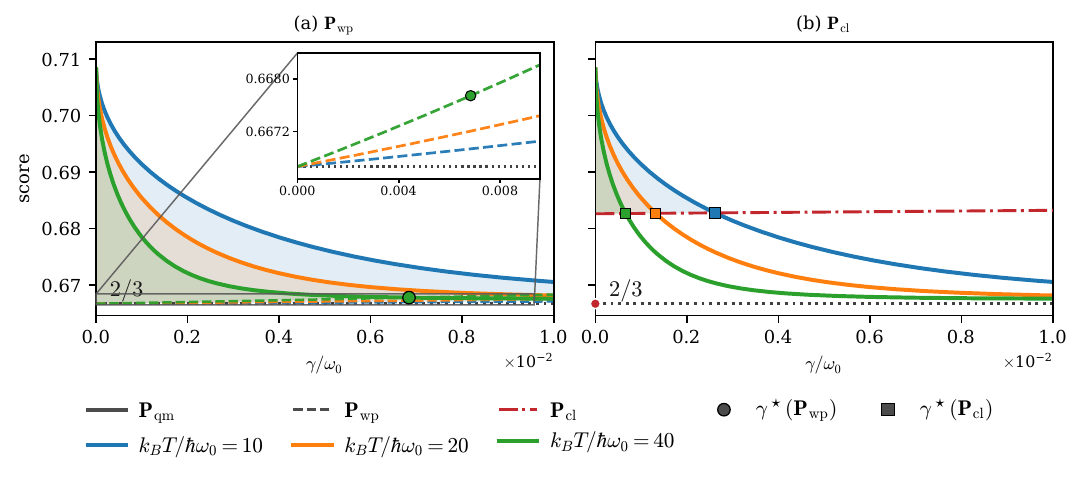}
    \caption{\justifying Score thresholds for the Caldeira-Leggett channel at $k_B T/\hbar\omega_0 = 10, 20, 40$ (colour). It compares the quantum maximum $\mathbf{P}_\text{qm} = \lambda_{\max}(\hat S)$ (thick solid,~\eqref{eq:qh-cl}) against (a) the Wigner positivity bound $\mathbf{P}_{\rm wp}$ (dashed,~\eqref{eq:bathtub bound}) and (b) against the temperature-independent Dirac bound $\mathbf{P}_\text{cl}$ (dash-dotted,~\eqref{eq:P3cl_cl}).}
    \label{fig:bounds_cl}
\end{figure*}

As for the thermal relaxation case, the Dirac bound jumps abruptly to a temperature-independent value at $\gamma\to 0^+$
\begin{align}\label{eq:Pcl_jump_CL}
    \mathbf{P}_{\rm cl}(\gamma\to 0^+, T) = \frac{1}{3}\bqty{1 + \Phi(\zeta_1) + \Phi(-\zeta_2)} \approx 0.6826
\end{align}
with $\zeta_1 \approx 1.100$ and $\zeta_2 \approx 0.903$ (\zcref{app:limits}).

The Wigner positivity bound derivation carries over similarly: it is continuous at $\gamma=0$ and rises with a slope $\propto \gamma k_B T$ (\zcref{fig:bounds_cl}),
\begin{align}\label{eq:wp_CL_0}
    \mathbf{P}_{\rm wp}(\gamma \to 0, T) = \frac{2}{3} +  3.71\times 10^{-3} \frac{k_B T}{\hbar \omega_0} \frac{\gamma}{\omega_0} + \mc{O}(\gamma^2)
\end{align}
See~\zcref{app:wp_thermal}.

Similar to the thermal relaxation case, lifting from Weyl symbols in eq.~\eqref{eq:mean_cl} to operator gives the damped quadrature
\begin{align}\label{eq:qh-cl}
        \hat q_H(t) := \ee^{-\gamma t/2}\bqty{\hat q\pqty{\cos\Omega t + \tfrac{\gamma}{2\Omega}\sin\Omega t} + \frac{\hat p}{m\Omega}\sin\Omega t}, 
\end{align}
which gives the quantum positivity projector
\begin{align}\label{eq:posq_op_cl}
    \pos(\hat q; t) = \Phi\!\pqty{\frac{\hat q_H(t)}{\sigma_q(t)}},
\end{align}
used to define the score operator~\eqref{eq:Q_def}; the maximum quantum score is then found numerically.

\subsection{Result}
Figure~\ref{fig:bounds_cl} shows the relevant bounds for the Caldeira--Leggett channel at $k_B T/\hbar\omega_0 = 10$, $20$ and $40$, all inside the high-temperature regime of the Caldeira-Leggett master equation~\eqref{eq:caldeira_leggett}. As in the thermal relaxation channel, the Dirac bound jumps to $\approx 0.683$ and crosses $\mathbf{P}_\text{qm}$ early, at: 
$\gamma^\star/\omega_0 \approx 2.61\times 10^{-3}$ ($k_B T/\hbar\omega_0 = 10$), $1.32\times 10^{-3}$ ($k_B T/\hbar\omega_0 = 20$), $6.63\times 10^{-4}$ ($k_B T/\hbar\omega_0 = 40$).
The scale-aware bound $\mathbf{P}_{\rm wp}$ leaves the noiseless $2/3$ bound continuously for increasing noise (\eqref{eq:wp_CL_0}), keeping the certifiable region open across the plotted window $\gamma/\omega_0 \leq 10^{-2}$ for $k_B T/\hbar\omega_0 = 10$ and $20$ and closing it only at $\gamma^\star(\mathbf{P}_{\rm wp}) \approx 6.84\times 10^{-3}$ for $k_B T/\hbar\omega_0 = 40$. 
The temperature still compresses the certifiable region: both $\mathbf{P}_{\rm wp}$ rises and $\mathbf{P}_\text{qm}$ decays faster at higher $T$.
However, against the physical Wigner positivity threshold the certifiable region is several times wider than against the Dirac bound at $k_B T/\hbar \omega_0 = 40$, and
open across the plotted window at the two lower temperatures.

Again, experimentally this range of $\gamma/\omega_0$ is applicable to e.g. ion trapping and circuit/cavity QED, where, compared to the thermal relaxation of the previous section, higher temperatures are considered, as shown in~\zcref{fig:bounds_cl}. 

\section{Case 3: Pure dephasing}\label{sec:fpe_deph}
Phase noise without energy exchange is modelled by Lindblad master equation with jump operator $\hat{L}_\phi = \sqrt{\gamma_\phi}\hat{a}^\dag\hat{a}$. 
The dissipator $\mathcal{D}$ will take the form of a nested commutator: $\mathcal{D}[\hat n]\,\cdot\, = -\tfrac{1}{2}[\hat n,[\hat n,\,\cdot\,]]$, resulting in the Lindblad equation:
\begin{equation}\label{eq:lindblad_dephasing}
    \dot{\hat\rho} 
    = -\frac{\ii}{\hbar}[\hat H_0, \hat\rho] + \gamma_\phi\,\mathcal{D}[\hat n]\hat\rho 
    = -\frac{\ii}{\hbar}[\hat H_0, \hat\rho] - \frac{\gamma_\phi}{2}[\hat n,[\hat n,\,\hat\rho\,]].
\end{equation}
Each noise commutator maps to a Moyal bracket with the degree-2 Weyl symbol $n$, which truncates exactly at the Poisson bracket because $\partial^3 n = 0$ (see~\zcref{app:wigner_weyl}); the Wigner equivalent of~\eqref{eq:lindblad_dephasing} therefore closes at finite order in $\hbar$.

After rescaling to dimensionless
quadratures $[Q,P]=\ii$, $(Q, P) = (r\cos \theta, r\sin \theta)$,the Wigner Fokker--Planck equation takes the form (see \zcref{app:greens}):
\begin{equation}\label{eq:fpe_deph}
\frac{\partial W}{\partial t} = \omega_0\frac{\partial W}{\partial \theta} +\frac{\gamma_\phi}{2}\,\frac{\partial^2 W}{\partial\theta^2}.
\end{equation}
The radial profile is frozen and the angular center evolves as $\theta(t) = \theta_0 - \omega_0 t$, with diffusion factor $\gamma_\phi/2$.
The dynamics on phase space distributions is the classical Fokker--Planck equation of action-conserving phase diffusion.
This is illustrated in fig.~\ref{fig:phase_space_evolution}.

\subsection{Classical \& Quantum Scores}\label{sec:dephasing_classical}
The Green's function of~\eqref{eq:fpe_deph} for an initial Dirac located at $(r_0, \theta_0)$ given the dynamics in~\eqref{eq:fpe_deph} is
\begin{equation}\label{eq:green_dephasing}
    G_\phi(r, \theta, t \mid r_0, \theta_0) = \frac{\delta(r-r_0)}{2\pi r_0} \sum_{m\in \mathbb{Z}} \ee^{\ii m(\theta - \phi_t)}\, \ee^{-\gamma_\phi m^2 t/2}.
\end{equation}
It is centred on the precessing angle $\phi_t \equiv \theta_0 - \omega_0 t$. Substituting~\eqref{eq:green_dephasing} into the Dirac trajectory score~\eqref{eq:P3_dirac} and integrating over the half-plane $\{q > 0\} = \{\cos\theta > 0\}$ at fixed $r = r_0 > 0$ yields a Fourier series, which is independent of $r_0$:
\begin{align}
    S(\zz_0) = \frac{1}{2} + \frac{2}{3\pi}\sum_{k=0}^{2}\sum_{n=0}^{\infty}&\frac{(-1)^n}{2n+1}\,\ee^{-\gamma_\phi(2n+1)^2 t_k/2} \nonumber \\ &\cos\,\!\bigl((2n+1)\phi_{t_k}\bigr), \label{eq:P3cl_dephasing}
\end{align}

At $\gamma_\phi = 0$, the Fourier series reconstructs $\Theta(\cos(\phi_{t_k}))$ via the standard expansion of the square wave, recovering the noiseless geometric value $2/3$. Including the noise, the classical bound depends only on $\theta_0$:
\begin{equation}\label{eq:Pcl_dephasing}
    \mathbf{P}_\text{cl}(\gamma_\phi) = \sup_{\theta_0 \in [0, 2\pi)} \mc{S}[\delta_{\theta_0}].
\end{equation}
The classical bound is therefore evaluated by a one-dimensional numerical maximisation.

The radial independence has an important physical consequence. A coherent state $\ket{\alpha}$, whose Wigner function is a Gaussian of variance $1/2$, becomes angularly localised at large amplitude. Let $\theta^{*}$ be the optimal angle in \eqref{eq:Pcl_dephasing}. The coherent state $\ket{|\alpha|\ee^{i \theta^{*}}}$ approaches the angular Dirac delta $\delta(|\alpha|, \theta^{*})$ in the limit $\abs{\alpha} \to \infty$, and saturates~\eqref{eq:Pcl_dephasing} from below. Thus, the supremum is asymptotically attained by coherent states, which are Wigner-positive, and the certification thresholds for Wigner negativity ($\mathbf{P}_\text{wp}$) and nonclassicality ($\mathbf{P}_\text{cl}$) coincide for this channel. Consequently, the certifiable region of the dephasing channel never strictly closes at any noise strength; what decays with $\gamma_\phi$ is the violation margin but not the certifiability. 

Since $\hat L_\phi = \sqrt{\gamma_\phi}\hat n$ and $\hat H_0$ are simultaneously diagonal in the Fock basis, the adjoint dynamics is diagonal there and the score operator is obtained from $\pos(\hat q)$ by exponential suppression of off-diagonal matrix elements:
\begin{equation}\label{eq:posq_dephasing}
    \mel{m}{\pos(\hat q; t)}{n} = \ee^{\ii\omega_0(m-n)t}\, \ee^{-\gamma_\phi(m-n)^2 t/2}\, \mel{m}{\pos(\hat q)}{n}.
\end{equation}

\subsection{Result}
\zcref[S]{fig:bounds_dephasing} compares the bounds for the dephasing channel over the small-noise window $\gamma_\phi/\omega_0 \in [0, 0.05]$. In contrast to the thermal channels, the classical bound rises continuously and only marginally above the noiseless bound $2/3$.

Dephasing is be caused by changes to the the resonance frequency, experimentally this can be sourced by for example fluctuating charge noise (shifting the trap potential or cavity frequency), voltage noise in ion traps, or magnetic field fluctuations in dielectric/superconducting traps.
Typically, the thermal noise is dominant over the dephasing channel.

\begin{figure}
    \centering
    \includegraphics[width=\linewidth]{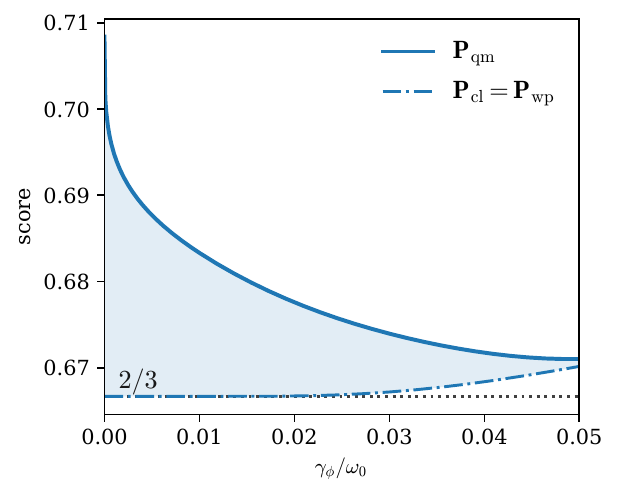}
    \caption{\justifying Quantum maximum score ($\mathbf{P}_{\rm qm}$) and Dirac-delta bound (which coincides with the exact Wigner positivity bound in this case $\mathbf{P}_{\rm cl}= \mathbf{P}_{\rm wp,e} = \mathbf{P}_{\rm wp}$) for the pure dephasing channel}
    \label{fig:bounds_dephasing}
\end{figure}

\section{Methods \& Results}\label{sec:results}
The previous sections have given an condensed outline of the derivation of the quantum and classical bounds. 
The supporting analytics for the score operator, Weyl symbols and dynamics in phase space can be found in~\zcref{app:score_op_matrix,app:wigner_weyl,app:greens}, respectively. 
Additionally, the thermal channels are discussed in more detail in~\zcref{app:limits,app:wp_thermal}, concerning the discontinuity in the Dirac bound and derivation of the Wigner positivity bounds, respectively. How a bounded frequency uncertainty propagates into the certification thresholds is discussed in~\zcref{app:freq_shift}.

Besides the supporting analytics in the appendix, we here give a brief description of the used numerical methods. Additionally, we point out some interesting features of the result regarding the noiseless optimal state and the certifiable region. 

\begin{figure*}[hbt!]
    \centering
    \includegraphics[width=\linewidth]{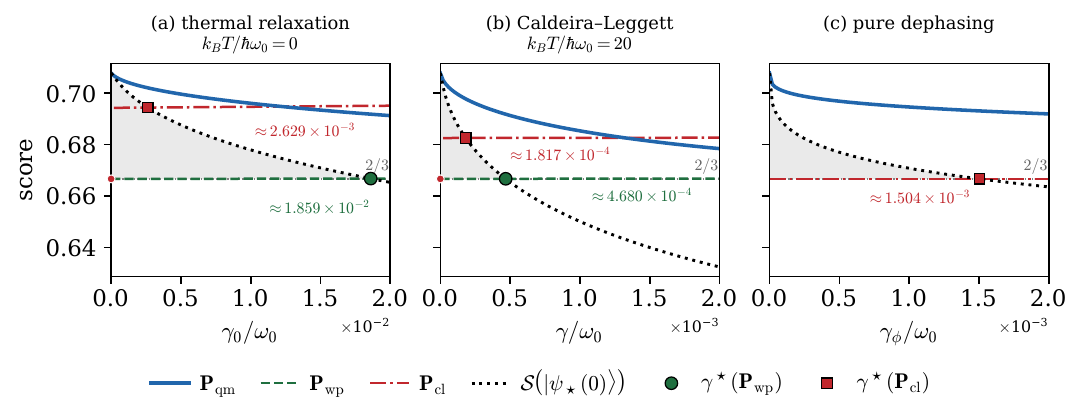}
    \caption{\justifying Robustness of noiseless optimal state the three noise models discussed here: (a) thermal relaxations of~\zcref{sec:photon-fluctuation}, (b) pure dephasing of~\zcref{sec:fpe_deph} and (c) the Caldeira--Leggett model of~\zcref{sec:fpe_cl}. 
    The figure plots the maximum quantum score (solid), the Wigner positivity bound (dashed) and the Dirac classical bound (dash-dotted), as well as the projection of noisy score operator on the noiseless optimal state $\ket{\psi_\star(0)}$~\eqref{eq:noisy-optimal}. For pure dephasing $\mathbf{P}_{\rm wp} = \mathbf{P}_{\rm cl}$.}
    \label{fig:noiseless_robustness}
\end{figure*}
\begin{figure*}[hbt!]
    \centering
    \includegraphics[width=0.8\linewidth]{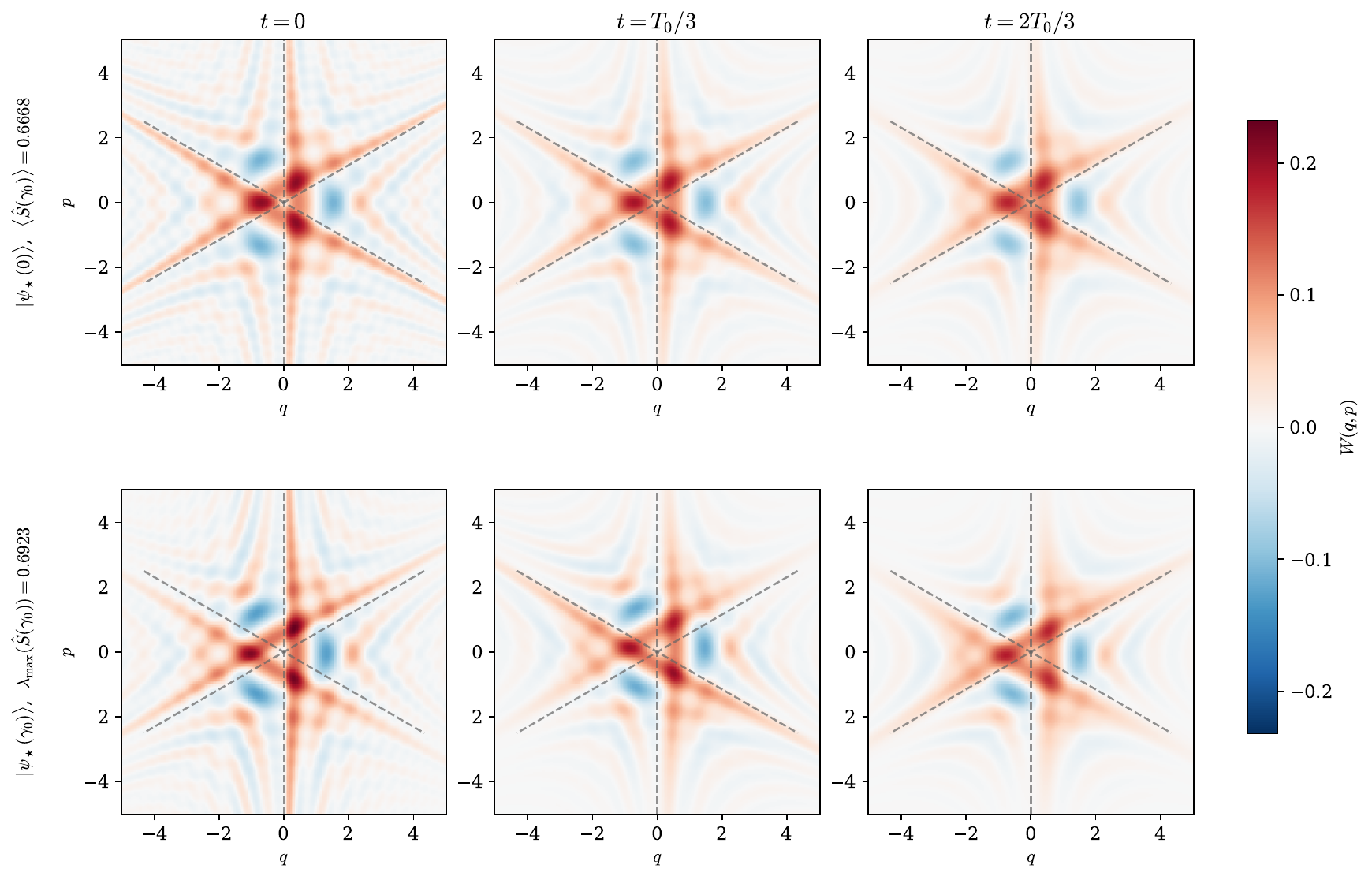}
    \caption{\justifying Evolution of noiseless (top panel) and adapted (bottom panel) optimal state in the thermal relaxation model at $\gamma_0/\omega_0 = 1.8\times 10^{-2}$}
    \label{fig:optimal_state_evol}
\end{figure*}

\subsection{Finding the quantum score}
Unlike the unitary dynamics in the ideal case, the positivity projector does not commute with
the dissipative evolution. Under a noisy channel, the evolved positivity projector $\ee^{\mc L^\dagger t}[\pos(\hat q)]$ is not the same as the projector onto the positive part of the evolved quadrature $\pos(\hat q(t))$. It is not clear, then, that the score operator possesses an elementary expression, and one is left propagating the full adjoint Liouvillian, which is costly and limits our numerical reach. For example, propagating the vectorized adjoint Liouvillian via the standard QuTiP mesolve routine is limited to a Fock cut-off of order $D\lessapprox 100$ in practice for the numerics done in this paper. We avoid that propagation by passing through phase space. The Weyl symbol of $\pos(\hat q)$ is the Heaviside step $\Theta(q)$, and for the channels treated here, the adjoint dynamics acts on Weyl symbols as a backward Fokker-Planck flow determined by the convolution with the corresponding Green's function. That Green's function is an explicit Gaussian in the two thermal channels (\zcref{app:greens}), so the convolution is evaluated in closed form. Requantizing the resulting symbols returns the smoothed quadrature operator (\zcref{eq:posq_op_photon, eq:posq_op_cl}). 

The quantum maximum scores in~\zcref{fig:bounds_thermal,fig:bounds_cl,fig:bounds_dephasing} are obtained by diagonalizing the Fock-truncated score operator $\hat{S}_D := P_D \hat{S} P_D$ at $D = 2100$ where $P_D := \sum_{k=0}^{D-1} \ketbra{k}$ is the projector on the first $D$ Fock states. Since the expectation value $\bra{\psi}\hat{S}\ket{\psi}$ only reads matrix elements within the support of $\ket{\psi}$, any eigenvector at cut-off $D$ is a legitimate normalized state and
\begin{align}
    \lambda_{\max}(\hat{S}_D) \leq \mathbf{P}_{\rm qm}
\end{align}
with equality as $D\to \infty$. The plotted quantum curves are therefore certified lower bounds on the true quantum maximum score, and the genuine certifiable regions are at least the ones read from the figures. Matrix elements of $\hat{S}_D$ are given in \zcref{eq:posq_dephasing} for the pure dephasing case, while the more elaborate constructions for thermal models are deferred to \zcref{app:score_op_matrix}. 

\begin{figure*}[hbt!]
    \centering
    \includegraphics[width=0.8\linewidth]{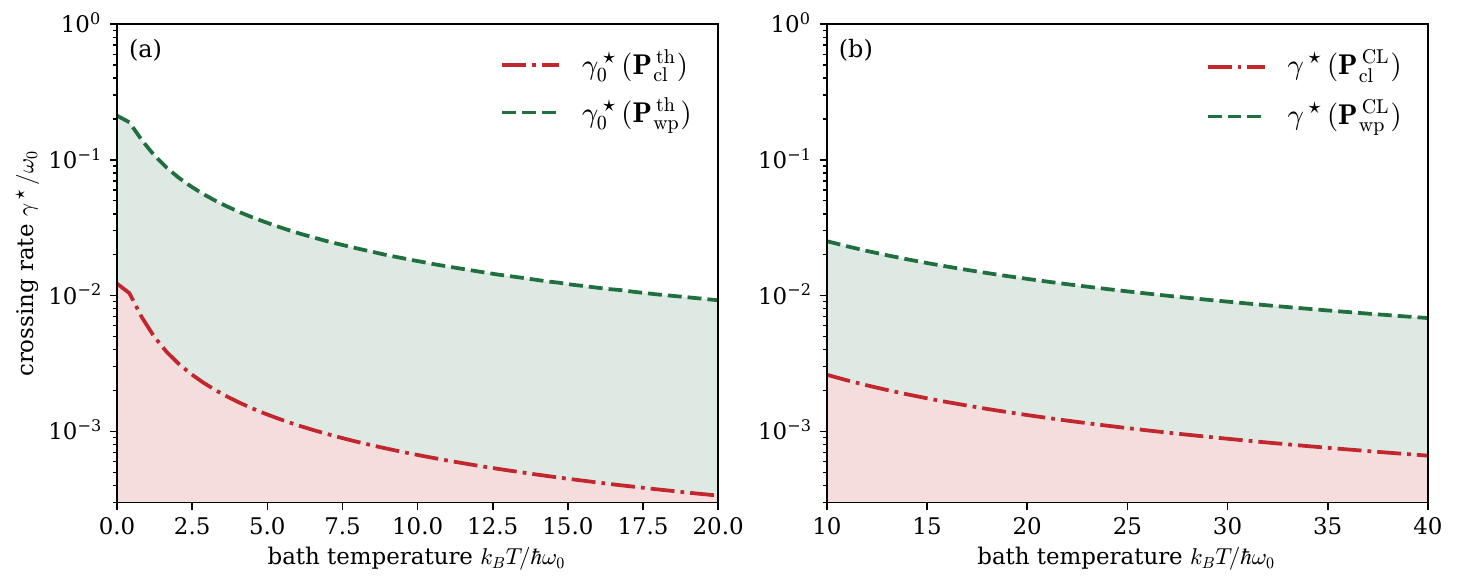}
    \caption{\justifying Certifiable region for thermal relaxation and Caldeira-Leggett model as functions of temperature.}
    \label{fig:gamma*_vs_T}
\end{figure*}

\subsection{Robustness of the noiseless optimum}
A realistic experiment may use the state optimized for the noiseless protocol ~\cite{zaw2022}, denoted $\ket{\psi_\star(0)}$, in a noisy device. As decoherence, which manifests as phase space diffusion, grows, it erodes the Wigner negativity responsible for the quantum violation and thus degrades the score of the noiseless optimal state (top panel in \zcref{fig:optimal_state_evol}). 

The score of this state under noise is given by the expectation value with the noise-dressed score operator:
\begin{equation}\label{eq:noisy-optimal}
    \mc{S}(\ket{\psi_\star(0)}, \gamma) = \mel{\psi_\star(0)}{\hat S(\gamma)}{\psi_\star(0)}\,,
\end{equation} 
\zcref[S]{fig:noiseless_robustness} compares this score against the classical Dirac and Wigner positivity bounds and the optimal quantum maximum $\mathbf{P}_{\rm qm}(\gamma)$. The non-adapted state degrades monotonically and ceases to certify well before the re-optimized one.
The crossing are $\gamma_0/\omega_0 = 2.629\times 10^{-3}$ against $\mathbf{P}_{\rm cl}$ and $1.859\times10^{-2}$ against $\mathbf{P}_{\rm wp}$ for thermal relaxation, $\gamma/\omega_0 = 1.817\times 10^{-4}$ and $4.680\times10^{-4}$ for the corresponding Caldeira-Leggett comparison, 
and at $\gamma_\phi \approx 1.504\times 10^{-3}$ under pure dephasing, where the two thresholds coincide.

\zcref[S]{fig:noiseless_robustness} thus shows that the noiseless optimal state is not the optimal state under noise. The adapted optimal state is defined by the same eigenvalue problem $\hat S(\gamma) \ket{\psi_\star(\gamma)} = \mathbf{P}_{\rm qm}\ket{\psi_\star(\gamma)}$. \zcref[S]{fig:optimal_state_evol} shows the two states side-by-side. The adapted optimal state is a deformation of the noiseless one with the negativity redistributed optimally for the diffused dynamics.

\subsection{Temperature dependence of certifiable region in thermal models}\label{sec:temp_dep}
As shown in~\zcref{fig:bounds_thermal,fig:bounds_cl}, increasing the temperature of the bath reduces the certifiable region. Here, we briefly note on the relative effects of the dissipation rate and bath temperature.
Both thermal relaxation and Caldeira--Leggett carry an additional temperature parameter: $\bar n$ and $k_B T$, respectively, related by $\bar n = \bqty{\exp\pqty{\frac{\hbar \omega}{k_BT}} - 1}^{-1}$. 
They control the diffusion strength independently of the dissipation rate. Increasing temperature accelerates the loss of quantum violation through decoherence and the certifiable region shrinks monotonically with these temperature parameters. 
\zcref[S]{fig:gamma*_vs_T} traces, for both thermal models, the crossing rates $\gamma_0^\star$ that bound the certifiable region for the two classical bounds. Meanwhile, the classical Dirac bound is temperature independent because it is agnostic to phase space scale.
The figure shows the importance of the bath temperature rather than the dissipation rate. Typically, only the combined $\gamma_0\bar{n}$ is measurable in (ion trapping) experiments. Only at sub-mK bath temperatures is it expected to be possible to measure $\gamma_0$ and $\bar{n}$ separately~\cite{brownnutt_ion-trap_2015}.

\section{Conclusion}\label{sec:conclusion}

We have extended Tsirelson's precession protocol from the isolated harmonic oscillator to an open-system setting, covering three standard dissipative channels relevant to continuous-variable hardware: thermal relaxation, pure dephasing, and Caldeira--Leggett quantum Brownian motion. Using the Moyal--Wigner phase space formalism, each channel reduces to a Fokker--Planck equation that is linear and positivity-preserving on the Wigner function. This allowed us to compute, in closed form, both the dissipation-dependent classical bounds and the maximum quantum score $\mathbf{P}_\text{qm}(\gamma) = \lambda_{\max}(\hat S)$ of the analytically derived score operator $\hat S$.
These bounds give a sufficient but not necessary requirement for non-classicality.
We have considered two notions of nonclassicality, each giving different comparison bounds: $\mathbf{P}_\text{cl}(\gamma)$ is the maximum score attainable by any non-negative phase space distribution, and $\mathbf{P}_\text{wp}(\gamma)$ is an upper bound for the maximum score attainable by any Wigner-positive function.
The former is attained by a Dirac extreme point, giving the strictest constraint for classicality, while violation of the latter certifies Wigner negativity. 
With the introduction of dissipation in the dynamics we have shown that (1) these two bounds no longer coincide for the thermal channels while they coincide as a suprema for pure dephasing, and (2) the comparison bounds change as a function of noise. 
The framework developed in this work can also test other notions of non-classicality, such as non-Gaussianity~\cite{Note1}, providing a flexible and efficient witness for various forms of non-classicality under noisy dynamics.

For small coupling to the environment, the precession protocol as a Wigner negativity witness is robust against noise.
However, the classical Dirac bound for classicality jumps up with the introduction of noise, making the quantum certifiable region very small even at low noise, especially under thermal relaxation.
Applications of the Tsirelson precession protocol should take into account their expected noise budget when constructing the dynamical nonclassicality witness.

\section*{Acknowledgments}

We thank Lin Htoo Zaw and Jackson Tiong for several discussions.

This work is supported by the National Research Foundation, Singapore through the National Quantum Office, hosted in A*STAR, under its Centre for Quantum Technologies Funding Initiative (S24Q2d0009).

\bibliography{references}
\newpage

\appendix
\begin{widetext}


\section{Weyl symbols, Wigner function, and Moyal star product}\label{app:wigner_weyl}
For an operator $\hat F$ acting on $L^2(\mathbb R)$, the Weyl symbol is \cite{curtright2013}
\begin{equation}\label{eq:weyl_symbol_def}
    F(q,p) = \int \dd y\,\ee^{-\ii p y/\hbar}\,\mel{q + y/2}{\hat F}{q - y/2},
\end{equation}
and the inverse (Weyl quantisation) reconstructs $\hat F$ from $F$ via the parity-displacement kernel. The Wigner function of $\hat\rho$ is the symbol of $\hat\rho/(2\pi\hbar)$,
\begin{equation}\label{eq:wigner_def}
    W_\rho(q,p) = \frac{1}{2\pi\hbar}\int \dd y\,\ee^{-\ii p y/\hbar}\,\mel{q+y/2}{\hat\rho}{q-y/2},
\end{equation}
which is real, normalised, and reproduces the correct quadrature marginals: $\int W_\rho\,\dd p = \mel{q}{\hat\rho}{q}$ and $\int W_\rho\,\dd q = \mel{p}{\hat\rho}{p}$. Under this map the operator product translates to the Moyal star product
\begin{equation}\label{eq:star_app}
    (F\star G)(q,p) = F(q,p)\,\exp\!\bqty{\tfrac{\ii\hbar}{2}\pqty{\overleftarrow{\partial}_q\overrightarrow{\partial}_p - \overleftarrow{\partial}_p\overrightarrow{\partial}_q}}\,G(q,p),
\end{equation}
with derivative expansion
\begin{align}\label{eq:star_expansion}
    F\star G =& FG + \tfrac{\ii\hbar}{2}\PB{F}{G} 
    - \tfrac{\hbar^2}{8}\bqty{\partial_q^2 F\,\partial_p^2 G - 2\partial_q\partial_p F\,\partial_q\partial_p G + \partial_p^2 F\,\partial_q^2 G} 
    + \mathcal O(\hbar^3),
\end{align}
where $\PB{F}{G} = \partial_q F\,\partial_p G - \partial_p F\,\partial_q G$ is the Poisson bracket. The series terminates exactly whenever one of the symbols is a polynomial of degree $\leq 2$ in $(q,p)$, since all third and higher phase space derivatives of a quadratic symbol vanish. The Moyal bracket
\begin{equation}\label{eq:moyal_bracket_app}
    \MB{F}{G} = \frac{F\star G - G\star F}{\ii\hbar} = \PB{F}{G} + \mathcal O(\hbar^2)
\end{equation}
inherits the same termination, and $\MB{F}{G} = \PB{F}{G}$ exactly whenever $F$ or $G$ is at most quadratic~\cite{dubois2021,hernandez2024}. 

For the Fock-basis calculation in all of the following appendices, we rescale the canonical variables to dimensionless quadratures $\hat Q = \sqrt{m\omega_0/\hbar}\,\hat q$, $\hat P = \hat p/\sqrt{m\hbar\omega_0}$
satisfying $[\hat Q, \hat P]=\ii$ for ease of analytic derivation. Passing to the dimensionless quadratures effectively rescales $\hbar \to 1$: the star prefactor becomes $\ii/2$, and the Moyal bracket reads $\MB{F}{G} = (F\star G - G\star F)/\ii$.

\section{Matrix elements of the score operator for the thermal models}\label{app:score_op_matrix}
The protocol's quantum score is obtained through diagonalizing the score operator $\hat S$ in the truncated Fock basis. This appendix details the closed-form matrix element of $\hat{S}$ for the thermal models considered in this paper, which share a common structure. At $t_0 = 0$, the projector is undamped and reduces to the positivity projector
\begin{align}
    \text{pos}(\hat{Q}, t_0) := \frac{1}{2}[\mathbbm{1} + \sgn(\hat{Q})]
\end{align}
where the sign operator $\sgn(\hat Q)$ has matrix elements that are zero if the indices have the same parity, and, for $m$ even and $n$ odd takes the closed form~\cite{zaw2022}
\begin{align}
\bra{m}\sgn(\hat{Q})\ket{n} = \frac{(-1)^{(n-m-1)/2}2^{-[(n+m)/2-1]}}{n-m}\sqrt{\frac{n}{\pi}\binom{m}{\frac{m}{2}}\binom{n-1}{\frac{n-1}{2}}}
\end{align}
The remaining elements ($m$ odd, $n$ even) are fixed by the Hermiticity of $\sgn(\hat Q)$. 

At later probing times $t_1, t_2$ the Heisenberg-evolved projector is a smoothed step of a damped and rotated quadrature 
\begin{align*}
    \text{pos}(\hat{Q}, t) = \Phi(a_t \hat{Q} + b_t \hat{P}) 
\end{align*}
whose matrix elements
\begin{align}\label{}
    \bra{m}\text{pos}(\hat{Q}, t)\ket{n} = \bra{m}\Phi(a_t \hat{Q} + b_t \hat{P})\ket{n} = \ee^{i(m-n)\beta_t}\bra{m}\Phi(c_t \hat{Q})\ket{n}
\end{align}
where $a_t+\ii b_t =: c_t \ee^{i\beta_t}$. 

In the following, we derive recursion relations for the Fock-basis
matrix elements of $\Phi(c\hat Q)$. We start with the following lemma:
\begin{lemma}\label{lemm:comm_f}
    Let $f$ be an entire function whose convergent power series $f(x) = \sum_{k=0}^\infty c_kx^k$. Then,
    \begin{align*}
        [\hat a^\dagger, f(\hat Q)] = -\frac{1}{\sqrt{2}}f'(\hat Q)
    \end{align*}
\end{lemma}
\begin{proof}
By induction, $[\hat a^\dagger, \hat Q^k] = -\frac{k}{\sqrt{2}} \hat{Q}^{k-1}$. Applying this for each term in the convergent power series gives
\begin{align*}
    [\hat a^\dagger, f(\hat Q)]  = \sum_{k=0}^\infty c_k [\hat a^\dagger, \hat Q^k] = -\frac{1}{\sqrt{2}}\sum_{k=0}^\infty c_k k \hat Q^{k-1} = -\frac{1}{\sqrt{2}}f'(\hat Q)
\end{align*}
\end{proof}
Recall that $\frac{\dd}{\dd x}\Phi(cx) = c\phi(cx), \ \frac{\dd}{\dd x}\phi(cx) = -c^2x\phi(cx)$ and apply \zcref{lemm:comm_f}, we have
\begin{align}\label{eq:lemm1_app_1}
   [\hat a^\dagger, \Phi(c\hat Q)] &= -\frac{c}{\sqrt{2}} \phi(c\hat Q)\\
    [\hat a^\dagger, \phi(c\hat Q)]&= \frac{c^2}{\sqrt{2}}\hat{Q}\phi(c\hat Q)\label{eq:lemm1_app_2}
\end{align}
Denote $A_{m,n}:=  \bra{m}\Phi(c\hat Q) \ket{n}$ and $B_{m,n} =  \bra{m}\phi(c\hat Q) \ket{n}$. In Fock basis, both matrices are real and
symmetric. Parity, together with $\Phi(-x)=1-\Phi(x)$ and $\phi(-x)=\phi(x)$, gives  
\begin{align*}
    A_{mm}&=\frac{1}{2}\\
    A_{mn}&=0
        \quad\text{if }m\neq n\text{ and }m+n\text{ is even},\\
    B_{mn}&=0
        \quad\text{if }m+n\text{ is odd}.
\end{align*}

For $n\geq1$, taking the matrix elements of
\zcref{eq:lemm1_app_1, eq:lemm1_app_2} between $\bra m$ and $\ket{n-1}$ gives
\begin{align*}
    \sqrt{m}\,A_{m-1,n-1}-\sqrt{n}\,A_{mn}
        &=-\frac{c}{\sqrt{2}}B_{m,n-1},\\
    \sqrt{m}\,B_{m-1,n-1}-\sqrt{n}\,B_{mn}
        &=\frac{c^2}{2}
          \left(
              \sqrt{n-1}\,B_{m,n-2}
              +\sqrt{n}\,B_{mn}
          \right),
\end{align*}
where we used
$[\hat Q,\phi(c\hat Q)]=0$ in the second line. Solving these
relations yields
\begin{align}
    A_{mn}
        &=\sqrt{\frac{m}{n}}A_{m-1,n-1}
          +\frac{c}{\sqrt{2n}}B_{m,n-1},\\
    B_{mn}
        &=r\sqrt{\frac{m}{n}}B_{m-1,n-1}
          -(1-r)\sqrt{\frac{n-1}{n}}B_{m,n-2}\end{align}
for $0\leq m \leq n$ and $r:=\left(1+\frac{c^2}{2}\right)^{-1}$.

The required seed values are $A_{00}=\frac{1}{2}$ and $B_{00} =\frac{1}{\sqrt{2\pi}}
          \bra{0}\ee^{-c^2\hat Q^2/2}\ket{0}
          =\sqrt{\frac{r}{2\pi}}$.

The recurrences determine the upper triangles column by column. Any
term carrying a negative Fock index is set to zero, while an element
below the diagonal is replaced using
$A(B)_{mn}=A(B)_{nm}$. The lower triangles then follow
by symmetry.

\section{Dissipative dynamics in phase space}\label{app:greens}
We derive in this appendix the dynamics of Wigner function in phase space and the explicit Green's functions $G(\mathbf{z}, t \mid \mathbf{z}_0)$ used to evaluate~\eqref{eq:P3_dirac} for each of the three channels.

\subsection{Thermal relaxation}
\subsubsection{Phase space dynamics}
With $\hat a = (\hat Q + \ii\hat P)/\sqrt 2$, the Weyl symbols of $\hat a, \hat a^\dag$ are complex conjugates of each other in $(Q, P)$ :
\begin{equation}\label{eq:weyl_a}
    a(Q,P) = \tfrac{1}{\sqrt 2}(Q + \ii P), \qquad a^*(Q,P) = \tfrac{1}{\sqrt 2}(Q - \ii P).
\end{equation}
The jump operators $\hat L_- = \sqrt{\gamma_0(\bar n+1)}\,\hat a$ and $\hat L_+ = \sqrt{\gamma_0\bar n}\,\hat a^\dag$ therefore have linear symbols, and the FPE is exact.

With $\hat H_0 = \hbar\omega_0(\hat a^\dag\hat a + 1/2)$, the symbol is $H_0 = \tfrac{\omega_0}{2}(Q^2 + P^2)$. Since $H_0$ is quadratic, $\MB{H_0}{W} = \PB{H_0}{W} = \omega_0(Q\,\partial_P W - P\,\partial_Q W)$ exactly, generating harmonic precession.

For the dissipator terms, the relevant Poisson brackets are 
\begin{align}\label{eq:photon_pb}
    &\PB{Q}{a} = \tfrac{\ii}{\sqrt 2},\quad \ \  \PB{Q}{a^*} = -\tfrac{\ii}{\sqrt 2},\quad \\
    &\PB{P}{a} = -\tfrac{1}{\sqrt 2},\quad \PB{P}{a^*} = -\tfrac{1}{\sqrt 2}.
\end{align}

The contribution of $\hat L_-$ to the drift on $Q$, is
\begin{align*}
\mu_Q^{(-)} & = -\tfrac{\ii}{2}\,\gamma_0(\bar n+1)\bqty{a\,\PB{Q}{a^*} + a^*\,\PB{a}{Q}}             \\
& = -\tfrac{\ii}{2}\,\gamma_0(\bar n+1)\bqty{a\pqty{-\tfrac{\ii}{\sqrt 2}} - a^*\pqty{\tfrac{\ii}{\sqrt 2}}}\\ & = -\tfrac{\gamma_0(\bar n+1)}{2}\cdot\tfrac{a + a^*}{\sqrt 2}\\
& = -\tfrac{\gamma_0(\bar n+1)}{2}\,Q.
\end{align*}
Likewise $\hat L_+$ contributes $\mu_Q^{(+)} = +\tfrac{\gamma_0\bar n}{2}\,Q$, and the two combine into $\mu_Q^{\text{diss}} = -\tfrac{\gamma_0}{2}\,Q$. The identical calculation for $z_i = P$ gives $\mu_P^{\text{diss}} = -\tfrac{\gamma_0}{2}\,P$. Adding the Hamiltonian precession reproduces, in dimensionless variables, the drift $\mathbf A^{(Q,P)}\mathbf z = (-\gamma_0 Q/2 + \omega_0 P,\,-\omega_0 Q - \gamma_0 P/2)^T$. Substituting $Q = \sqrt{m\omega_0/\hbar}\,q$, $P = p/\sqrt{m\hbar\omega_0}$ converts this to the dimensionful drift matrix~\eqref{eq:drift_photon}.

The diffusion terms are given by
\begin{align*}
    D_{QQ} & = \tfrac{1}{2}\,\Re\!\bqty{\gamma_0(\bar n+1)\PB{Q}{a}\PB{Q}{a^*} + \gamma_0\bar n\,\PB{Q}{a^*}\PB{Q}{a}}
    = \tfrac{\gamma_0(2\bar n+1)}{4}, \\
    D_{PP} & = \tfrac{\gamma_0(2\bar n+1)}{2}\,\Re\!\bqty{\pqty{-\tfrac{1}{\sqrt 2}}\!\pqty{-\tfrac{1}{\sqrt 2}}} = \tfrac{\gamma_0(2\bar n+1)}{4},   \\
    D_{QP} & = \tfrac{\gamma_0(2\bar n+1)}{2}\,\Re\!\bqty{\tfrac{\ii}{\sqrt 2}\!\pqty{-\tfrac{1}{\sqrt 2}}} = 0,
\end{align*}
so $\boldsymbol D^{(Q,P)} = \tfrac{\gamma_0(2\bar n+1)}{4}\,\id$. Restoring the dimensions via $D_{qq} = (\hbar/m\omega_0)\,D_{QQ}$, $D_{pp} = m\hbar\omega_0\,D_{PP}$, $D_{qp} = \hbar\,D_{QP}$ recovers the dimensionful diffusion matrix~\eqref{eq:D0}.

\subsubsection{Green's function}
The Fokker--Planck equation~\eqref{eq:fpe} with linear drift~\eqref{eq:drift_photon} and constant diffusion~\eqref{eq:D0} admits a Gaussian Green's function. We consider the ansatz
\begin{align}\label{eq:gaussian_ansatz}
    G(\mathbf z, t\mid \mathbf z_0) = &\frac{1}{2\pi\sqrt{\det\mathbf C(t)}} 
  \exp\!\bqty{-\tfrac{1}{2}(\mathbf z - \bar{\mathbf z}(t))^T\,\mathbf C(t)^{-1}\,(\mathbf z - \bar{\mathbf z}(t))},
\end{align}
with $\bar{\mathbf z}(0) = \mathbf z_0$ and $\mathbf C(0) = \mathbf 0$. Substituting into~\eqref{eq:fpe} and matching powers of $(\mathbf z - \bar{\mathbf z})$ yields ODEs for the first two moments,
\begin{align}
    \dot{\bar{\mathbf z}} & = \mathbf A\bar{\mathbf z}, \label{eq:zbar_FP}                              \\
    \dot{\mathbf C}       & = \mathbf A\mathbf C + \mathbf C\mathbf A^T + 2\mathbf D, \label{eq:cov_FP}
\end{align}
The solution is given by
\begin{align}
    \bar{\mathbf z}(t) &= \ee^{\mathbf A t}\mathbf z_0\\
    \mathbf C(t) &= \int_0^t \ee^{\mathbf A s}(2\mathbf D)\ee^{\mathbf A^T s}\,\dd s \label{eq:cov_integral}
\end{align}
The drift matrix in \eqref{eq:drift_photon} decomposes as $\mathbf A = -\tfrac{\gamma_0}{2}\id + \omega_0 J$, with $J = \begin{pmatrix} 0 & 1/(m\omega_0) \\ -m\omega_0 & 0\end{pmatrix}$. Since $\id$ and $J$ commute, $\ee^{\mathbf A t} = \ee^{-\gamma_0 t/2}(\cos\omega_0 t\,\id + \sin\omega_0 t\,J) = \ee^{-\gamma_0 t/2}R(\omega_0 t)$, and we immediately arrive at~\eqref{eq:spiral_centre}.

A direct calculation gives $R(\phi)\,\boldsymbol D\,R(\phi)^T = \boldsymbol D$. Using this invariance, the covariance becomes
\begin{align}\label{eq:coC_{pp}hoton_appendix}
    \mathbf C(t) =& \ 2\boldsymbol D\int_0^t \ee^{-\gamma_0 s}\,\dd s\\
    =& \ \frac{2\boldsymbol D}{\gamma_0}\bigl(1 - \ee^{-\gamma_0 t}\bigr) \\
    =& \ \frac{\hbar(2\bar n + 1)}{2}\bigl(1 - \ee^{-\gamma_0 t}\bigr)\begin{pmatrix} 1/(m\omega_0) & 0 \\ 0 & m\omega_0 \end{pmatrix},
\end{align}
in agreement with~\eqref{eq:variance_photon}. The position marginal at time $t$ is Gaussian with mean $\bar q(t) = \ee^{-\gamma_0 t/2}[q_0\cos\omega_0 t + (p_0/m\omega_0)\sin\omega_0 t]$ and variance $\sigma_q^2(t) = \tfrac{\hbar(2\bar n + 1)}{2m\omega_0}(1 - \ee^{-\gamma_0 t})$.

\subsection{Dephasing}
\subsubsection{Phase space dynamics}
The dephasing jump operator $\hat L_\phi = \sqrt{\gamma_\phi}\,\hat n$ has Weyl symbol $L_\phi = \sqrt{\gamma_\phi}\bqty{\tfrac{1}{2}(Q^2 + P^2) - \tfrac{1}{2}}$, which is quadratic. The Wigner equation closes at finite order in $\hbar$ because $\hat L_\phi$ is Hermitian: the dissipator reduces to a nested commutator,
\begin{equation}\label{eq:dephasing_double_comm}
    \mathcal{D}[\hat n]\hat\rho = \hat n\hat\rho\hat n - \tfrac{1}{2}\acomm{\hat n^2}{\hat\rho} = -\tfrac{1}{2}[\hat n,[\hat n,\hat\rho]],
\end{equation}
so the full master equation involves only Moyal brackets with $\hat n$.

In the dimensionless convention $[Q,P]=\ii$, the number-operator Weyl symbol is
\begin{equation}\label{eq:number_symbol_dephasing}
    n(Q,P)=\frac{Q^2+P^2-1}{2}.
\end{equation}
The commutator maps as
\begin{equation}\label{eq:dephasing_commutator_map}
    [\hat n,\hat\rho]_W=\ii\MB{n}{W}=\ii\PB{n}{W},
\end{equation}
because $n$ is quadratic. Moreover,
\begin{equation}\label{eq:pb_n}
    \PB{n}{F}=Q\partial_PF-P\partial_QF=\partial_\theta F.
\end{equation}
Consequently the Hamiltonian and dephasing contributions are, respectively,
\begin{equation*}
    -\ii\omega_0[\hat n,\hat\rho]_W=\omega_0\partial_\theta W,
    \qquad
    -\frac{\gamma_\phi}{2}[\hat n,[\hat n,\hat\rho]]_W
    =\frac{\gamma_\phi}{2}\partial_\theta^2W.
\end{equation*}
Adding the two contributions reproduces the angular drift--diffusion equation~\eqref{eq:fpe_deph},
\begin{equation*}
    \partial_t W = \omega_0\,\partial_\theta W + \tfrac{\gamma_\phi}{2}\,\partial_\theta^2 W,
\end{equation*}
The radial profile is frozen because $\PB{n}{\,\cdot\,}$ depends only on $\theta$.

\subsubsection{Green's function}
Consider the Fourier modes $\ee^{\ii m\theta}$ ($m \in \mathbb Z$), which are eigenfunctions of both $\partial_\theta$ and $\partial_\theta^2$. Then,
\begin{equation*}
    \pqty{\omega_0\partial_\theta + \tfrac{\gamma_\phi}{2}\partial_\theta^2}\ee^{\ii m\theta} = \pqty{\ii m\omega_0 - \tfrac{\gamma_\phi m^2}{2}}\ee^{\ii m\theta},
\end{equation*}
The angular kernel is thus given by
\begin{equation}\label{eq:angular kernel}
    K(\theta, t\mid \theta_0) = \frac{1}{2\pi}\sum_{m \in \mathbb Z} \ee^{\ii m(\theta - \theta_0 + \omega_0 t)}\,\ee^{-\gamma_\phi m^2 t/2}.
\end{equation}
concentrated near $\theta_0 - \omega_0 t$ at small $t$ with angular variance $\gamma_\phi t$, relaxing to the uniform density $1/(2\pi)$ as $t\to\infty$. 

The two-dimensional Green's function is obtained by combining~\eqref{eq:angular kernel} with the frozen radial delta. Since the two-dimensional delta in polar coordinates carries a Jacobian $\delta^{(2)}(\mathbf z - \mathbf z_0) = \tfrac{1}{r_0}\delta(r - r_0)\delta(\theta - \theta_0)$, propagating only the angular part gives~\eqref{eq:green_dephasing},
\begin{equation*}
    G_\phi(r, \theta, t\mid r_0, \theta_0) = \frac{\delta(r - r_0)}{r_0}\cdot K(\theta, t\mid \theta_0).
\end{equation*}

\subsection{Caldeira-Leggett}
\subsubsection{Phase space dynamics}
The harmonic precession parts follow exactly the same as the thermal relaxation case. In the following, we utilize two results about Moyal star product to calculate the remaining dissipative terms:
\begin{enumerate}
    \item $\hat q$ has linear Weyl symbol $q$, so $q\star W = qW + \tfrac{\ii\hbar}{2}\partial_p W$ and $W\star q = qW - \tfrac{\ii\hbar}{2}\partial_p W$ exactly (the star expansion terminates at first order). Hence
          \begin{equation}\label{eq:comm_q_W}
              [\hat q, \hat\rho] \longleftrightarrow \ii\hbar\,\partial_p W,
          \end{equation}
          and similarly $[\hat p, \hat\rho] \leftrightarrow -\ii\hbar\,\partial_q W$.
    \item $\hat p$ has linear Weyl symbol $p$. The same termination gives $p\star W = pW - \tfrac{\ii\hbar}{2}\partial_q W$ and $W\star p = pW + \tfrac{\ii\hbar}{2}\partial_q W$, so the anticommutator translates simply as
          \begin{equation*}
              \acomm{\hat p}{\hat\rho} \longleftrightarrow p\star W + W\star p = 2pW.
          \end{equation*}
\end{enumerate}

The friction term can be calculated by combining~\eqref{eq:comm_q_W} with $\acomm{\hat p}{\hat\rho} \leftrightarrow 2pW$,
\begin{equation*}
    -\tfrac{\ii\gamma}{2\hbar}\bqty{\hat q, \acomm{\hat p}{\hat\rho}} \longleftrightarrow -\tfrac{\ii\gamma}{2\hbar}\cdot\ii\hbar\,\partial_p(2pW) = \gamma\,\partial_p(pW).
\end{equation*}
For the diffusion terms, applying~\eqref{eq:comm_q_W} twice,
\begin{equation*}
    -\tfrac{m\gamma k_B T}{\hbar^2}\bqty{\hat q, [\hat q, \hat\rho]} \longleftrightarrow -\tfrac{m\gamma k_B T}{\hbar^2}\cdot(\ii\hbar)^2\,\partial_p^2 W = m\gamma k_B T\,\partial_p^2 W.
\end{equation*}
The factor $\hbar^2$ from the double commutator cancels the $\hbar^{-2}$ in the master equation exactly, so the resulting FPE is $\hbar$-independent. 

Summing the three contributions,
\begin{equation*}
    \partial_t W = -(p/m)\partial_q W + m\omega_0^2 q\,\partial_p W + \gamma\,\partial_p(pW) + m\gamma k_B T\,\partial_p^2 W,
\end{equation*}
which is exactly the classical Klein-Kramer equations and takes the form~\eqref{eq:fpe} with linear drift $\mathbf A\mathbf z = (p/m,\,-m\omega_0^2 q - \gamma p)^T$  and diffusion $D_{pp} = m\gamma k_B T$ (all other entries zero) matching~\eqref{eq:drift_cl}.

\subsubsection{Green's function}
The drift~\eqref{eq:drift_cl} has eigenvalues $\lambda_\pm = -\gamma/2 \pm \ii\Omega$ with $\Omega = \sqrt{\omega_0^2 - \gamma^2/4}$. A direct diagonalisation gives the matrix exponential
\begin{equation}\label{eq:expA_cl}
    \ee^{\mathbf A t} = \ee^{-\gamma t/2}\begin{pmatrix} \cos\Omega t + \tfrac{\gamma}{2\Omega}\sin\Omega t & \sin\Omega t/(m\Omega) \\[3pt] -m(\omega_0^2/\Omega)\sin\Omega t & \cos\Omega t - \tfrac{\gamma}{2\Omega}\sin\Omega t \end{pmatrix},
\end{equation}
whose first row applied to $(q_0,p_0)^T$ reproduces the position trajectory~\eqref{eq:mean_cl}.

The covariance matrix can be solved analytically via \eqref{eq:cov_integral}. We are only interested in the component $C_{qq} = (\sigma_{q}^{\rm CL})^2$ given by the integral
\begin{align}\label{eq:sigma_CL_closed}
    C_{qq} &= 2m\gamma k_B T \int_{0}^t \bqty{\ee^{\mathbf{A}s}\text{diag}(0,1)\ee^{\mathbf{A}^\top s}}_{qq} \dd s\\
    &= 2m\gamma k_B T \int_{0}^t \bqty{\pqty{\ee^{\mathbf{A}s}}_{qp}}^2 \dd s\\
    &= 2m\gamma k_B T \int_{0}^t \bqty{\ee^{-\gamma s/2}\frac{\sin\Omega s}{m\Omega}}^2 \dd s \\
    &= \frac{k_BT}{m \omega_0^2}\Bqty{1 - \ee^{-\gamma t}\bqty{1 + \frac{\gamma}{2\Omega}\sin(2\Omega t) + \frac{\gamma^2}{2\Omega^2}\sin^2(\Omega t) }}
\end{align}

\section{Behavior of thermal classical bounds in  small and large noise regimes}\label{app:limits}
As mentioned in the main text,  the classical bound $\mathbf{P}_\text{cl}(\gamma)$ acquires a discontinuous jump at $\gamma \to 0^+$ in both thermal models and saturates at $5/6$ as $\gamma \to \infty$ for the thermal relaxation model. We derive those results here.

\subsection{Thermal relaxation}
For $\gamma_0 > 0$, the marginal at $t_0 = 0$ is a hard indicator ($\sigma_q(0) = 0$, so $\Phi$ collapses to the Heaviside function $\Theta$ at $t_0$). Meanwhile, due to diffusion, $\sigma_q(t_1, t_2) >0$. Differentiating along $Q_0$ at a fixed $P_0$, we get
\begin{align}\label{eq:dPdq0_photon_app}
    \frac{\partial  S}{\partial Q_0} <0
\end{align}
so the score function is strictly decreasing along $\{Q_0>0\}$ at a fixed $P_0$, and the half-plane $\{Q_0 <0\}$ contributes at most $2/3$ since $\Theta(Q_0) = 0$ there. The supremum over all Dirac functions $\delta(\zz_0)$ is consequently attained as a one-sided limit $\zz_0 \to (0^+, P_0)$. On this boundary, 
\begin{equation}\label{eq:boundary_sup_photon_app}
  \mathbf{P}_\text{cl}(0^+, P_0;\gamma_0,\bar n) = \tfrac{1}{3}\bigl[\,1 + \Phi(a_1 P_0) + \Phi(-a_2 P_0)\,\bigr],
\end{equation}
with
\begin{equation}\label{eq:def_ak_photon_app}
  a_k(\gamma_0,\bar n) := \frac{\sqrt 3}{2}\,\frac{\ee^{-\gamma_0 t_k/2}}{\sqrt{(\bar n+\tfrac{1}{2})(1-\ee^{-\gamma_0 t_k})}},
\end{equation}
for $k = 1, 2$.
Setting $\partial_{P_0} \mathbf{P}_{\rm cl} = 0$ in~\eqref{eq:boundary_sup_photon_app} gives $a_1 \phi(a_1 P_0) = a_2 \phi(a_2 P_0)$; taking logarithms,
\begin{equation}\label{eq:p0_star_photon_app}
  P_0^\star(\gamma_0, \bar n)^2 = \frac{2\,\ln(a_1/a_2)}{a_1^2 - a_2^2},
\end{equation}
which is finite and positive for every $\gamma_0 > 0$. The second derivative $\partial_{P_0}^2 S\big|_{P_0 = P_0^\star} = -\tfrac{1}{3}\bigl[a_1^3 P_0^\star \phi(a_1 P_0^\star) - a_2^3 P_0^\star \phi(-a_2 P_0^\star)\bigr]$ is strictly negative at the critical point (since $a_1 > a_2$ and $\phi$ is positive), so $P_0^\star$ is a maxima. Substituting back,
\begin{equation}\label{eq:Pcl_photon_closed_app}
  \mathbf{P}_\text{cl}(\gamma_0,\bar n) = \tfrac{1}{3}\Bigl[\,1 + \Phi(a_1 P_0^\star) + \Phi(-a_2 P_0^\star)\,\Bigr].
\end{equation}
At $\gamma_0 \to 0^+$, expanding $\ee^{-\gamma_0t/2} = 1- \gamma_0 t/2 + \mc{O}(\gamma_0^2)$ and $\sigma_q^2(t) = \pqty{\bar{n} + \frac{1}{2}}\gamma_0 t +\mc{O}(\gamma_0^2)$, we get
\begin{equation}
    a_k = \frac{\sqrt{3}/2}{\sqrt{\pqty{\bar{n} + \frac{1}{2}}\gamma_0 t_k}}[1+\mc{O}(\gamma_0)].
\end{equation}
The common factor $\sqrt{\pqty{\bar{n} + \frac{1}{2}}\gamma_0}$ cancels in the ratio $a_1/a_2 = \sqrt{t_2/t_1} = \sqrt{2}$. From~\eqref{eq:p0_star_photon_app},
\begin{equation}\label{eq:p0_star_limit_photon_app}
  \bigl[P_0^\star\bigr]^2 \xrightarrow{\gamma_0 \to 0^+} \frac{8\,\ln 2}{3}\,(\bar n + \tfrac{1}{2})\,\gamma_0 t_1 = \frac{16\pi\ln 2}{9}\,(\bar n + \tfrac{1}{2})\,\frac{\gamma_0}{\omega_0},
\end{equation}
so $P_0^\star \propto \sqrt{\pqty{\bar n + \tfrac{1}{2}}\frac{\gamma_0}{\omega_0}}$ collapses to the origin while the $\Phi$-arguments stay $\mathcal O(1)$:
\begin{equation}\label{eq:phi_args_photon_app}
  a_1 P_0^\star \to \sqrt{2\ln 2},\qquad a_2 P_0^\star \to \sqrt{\ln 2}.
\end{equation}  
The limiting value recovers the result in~\eqref{eq:Pcl_jump_photon} of the main text. 

In the limit $\gamma_0 \to \infty$, $\bar{q}_k \to 0$, $\sigma_q^2(t_k) \to \bar{n}+\frac{1}{2}$, and $a_k \to (\sqrt{3}/2)\ee^{-\gamma_0 t_k/2}/\sqrt{\bar{n}+\frac{1}{2}}$. The ratio $a_1/a_2 = \ee^{\gamma_0 T_0/6}$ diverges exponentially, so $a_1^2 - a_2^2 \sim a_1^2$. Eq.~\eqref{eq:p0_star_photon_app} then gives $P_0^\star \sim \sqrt{\pqty{\bar{n}+\frac{1}{2}}\gamma_0/\omega_0} \ee^{\gamma_0 T/6}$, and the $\Phi$-arguments then approach
\begin{equation}\label{eq:phi_args_inf_photon_app}
  a_1 P_0^\star \to \sqrt{\gamma_0 T_0/3} \to \infty,\qquad a_2 P_0^\star = a_1 P_0^\star\,\ee^{-\gamma_0 T_0/6} \to 0.
\end{equation}
Hence $\Phi(a_1 P_0^\star) \to 1$ and $\Phi(-a_2 P_0^\star) \to 1/2$, giving the $5/6$ saturation mentioned in the main text.

\subsection{Caldeira-Leggett}
The same boundary reduction applies: on the line $q_0 = 0^+$, the noiseless mean reads $\bar q^{\,\text{CL}}(t_k)\bigl|_{q_0=0} = \ee^{-\gamma t_k/2}\sin(\Omega t_k)\,p_0/(m\Omega)$. Defining
\begin{equation}\label{eq:def_ak_cl_app}
  a_k^{\,\text{CL}}(\gamma, T) := \frac{\ee^{-\gamma t_k/2}\,\abs{\sin(\Omega t_k)}}{m\Omega\,\sigma_q^{\,\text{CL}}(t_k)},\qquad k = 1, 2,
\end{equation}
the boundary score takes the same universal form~\eqref{eq:boundary_sup_photon_app} with $a_k \to a_k^{\,\text{CL}}$, and the closed-form maximiser is given by~\eqref{eq:p0_star_photon_app} with the same substitution.

Expanding the closed-form variance \eqref{eq:sigma_CL_closed} to the third order in $\gamma$ gives
\begin{align}\label{eq:sigma_CL_asymptote}
    \pqty{\sigma_{q}^{\rm CL}}^2 =& \frac{\gamma k_B T}{m\omega_0^3}\bqty{\tau(t) -\frac{\gamma}{\omega_0}h(t) + \mc{O}(\gamma^2)}
\end{align}
with (writing $\theta = \omega_0 t$) 
\begin{align}
    \tau &:= \theta -\frac{1}{2}\sin 2\theta\\
    h &:= \frac{1}{2}\pqty{\theta^2 -\theta \sin 2\theta+\sin^2\theta}
\end{align}
At probing times $t_k$:
\begin{align}
    \tau_1 = \frac{2\pi}{3}+\frac{\sqrt 3}{4}, \ \tau_2 = \frac{4\pi}{3} - \frac{\sqrt 3}{4}
\end{align}

Similar to the calculation for the thermal relaxation model, we find that
\begin{align}
    \frac{[\sigma_{q}^{\rm CL}(t_2)]^2}{[\sigma_{q}^{\rm CL}(t_1)]^2} \xrightarrow[]{\gamma\to 0^+} \frac{\tau_2}{\tau_1} \approx 1.486
\end{align}
which determines the $\Phi$-arguments in \eqref{eq:boundary_sup_photon_app}
\begin{align*}
    a_1p_0^\star \to \zeta_1 \approx 1.100, \ a_2 p_0^\star \to \zeta_2 \approx 0.903
\end{align*}
and reproduces the post-jump value \eqref{eq:Pcl_jump_CL} in the main text.

\section{Derivation of Wigner positive bound for thermal channels}\label{app:wp_thermal}
We use dimensionless variables $\ZZ_0 = (Q_0, P_0)$ under which the relaxed feasible set becomes $0\leq f(\ZZ)\leq \frac{1}{\pi}$ with $\int f(\ZZ) \dd \ZZ = 1$. 

The thermal relaxation and Caldeira–Leggett channels share the bathtub construction. For both, the Dirac trajectory score has the common form
\begin{align}
    S(\mathbf{Z}_0) = \frac{1}{3}\bqty{\Theta(Q_0) + \Phi(u_1)+ \Phi(u_2)}
\end{align}
where the $\Phi$-arguments are linear functions in $Q_0$ and $P_0$
\begin{align}
    u_k = a_k Q_0 + b_kP_0
\end{align}

For a threshold $\lambda$, rewriting the score and using the optimization constraints gives  
\begin{align}\label{eq:bathtub_lambda}
    \int S(\ZZ_0) f(\ZZ_0) \dd^2\ZZ_0  =& \ \lambda + \int [S(\ZZ_0) - \lambda] f(\ZZ_0) \dd^2\ZZ_0\\
    \leq & \  \lambda + \frac{1}{\pi} \int (S-\lambda)_+\dd^2\ZZ_0 =: B(\lambda) 
\end{align}
In the above expression and throughout this appendix, for notational convenience, we introduce the subscript $+$ to denote the integration over the set where the integrand is positive, that is,
\begin{align*}
  \int [f(x)]_+ \dd x := \int_{-\infty}^{\infty} \max(0, f(x)) \dd x 
\end{align*}

So, $B(\lambda)$ upper bounds $\mathbf{P}_{\rm wp}$ for every $\lambda$. The bound is saturated when $f$ matches the integrand pointwise: $f = 1/\pi$ when $S>\lambda$, $f=0$ where $S<\lambda$, and $f$ unconstrained in $[0, 1/\pi]$ where $S = \lambda$ -- the integrand vanishes there regardless of $f$. The maximizer is therefore given by 
\begin{align}\label{eq:bathtub_area}
    f^\star := \frac{1}{\pi}\bqty{\mathbbm{1}_{\{s>\lambda^\star\}} + c\mathbbm{1}_{\{S=\lambda^\star\}}}
\end{align}
with $|\{S>\lambda^\star\}|\leq \pi \leq |\{S\geq\lambda^\star\}|$, which is \eqref{eq:bathtub_optimizer} in the main text, and 
\begin{align}
    \mathbf{P}_{\rm wp} = \mc{S}[f^\star] = B(\lambda^\star) = \inf_{\lambda} B(\lambda)
\end{align}

Locating the optimal $\lambda^\star$ require solving the area condition \eqref{eq:bathtub_area}, which does not admit an analytic solution. We present here a bound for $\lambda^\star$ that is particularly useful for understanding the behavior of the bathtub bound in small-$\gamma_0$ regime. 

In the noiseless case, the score is the hard three-sector indicator of Tsirelson's construction, equal to $2/3$ on the sectors where two of the three half-planes overlap. The noise smooths the $t_{1,2}$ marginals through $\sigma_q(t_{1,2})>0$ but leaves $S\to 2/3^+$ as $|\ZZ_0| \to \infty $, so $|\{S>\lambda \}| = \infty$ for $\lambda \leq 2/3$. Therefore, $\lambda^\star >2/3$ and we can further show in the following that it lies exponentially close to $2/3$ from above (\zcref{eq:lambda_star_bound}), and $\mathbf{P}_{\rm wp}$ is exponentially close to $B(2/3)$ as a result (\zcref{eq:bathtub_bound_bound}). 

Let $\lambda = 2/3 + \epsilon$. We start by characterizing the set $\{S>\lambda\}$. For $Q_0< 0$, $\Theta(Q_0) = 0$, and so $S<2/3$ trivially here. For $Q_0 = 0$, $S = \frac{1}{3}\bqty{ \frac{1}{2} + \Phi(u_1)+\Phi(u_2)} <2/3$. We are left with $Q_0>0$, which gives $\Theta(Q_0) = 1$, and $u_2<-u_1/\kappa$ for $\kappa := -\frac{b_1}{b_2} > 1$. The excess set can thus be parametrized in $u_1$ and $u_2$ as
\begin{align}\label{eq:excess_set_parametrization}
    \{S> 2/3 + \epsilon\} = \{-\Phi^{-1}(\Phi(u_1) - 3 \epsilon) \leq u_2 \leq -u_1/\kappa\}
\end{align}
with area
\begin{align}\label{eq:excess_set_area_bound}
    |\{S> 2/3 + \epsilon\}| =& \ \mc{J}\int_{0}^{\infty} \pqty{ \Phi^{-1}(\Phi(u) - 3 \epsilon) -\frac{u}{\kappa}}_+
    \dd u \\
     \leq & \ \mc{J}\int_{0}^{\infty} \bqty{\min(u, \Phi^{-1}(1-3\epsilon)) - \frac{u}{\kappa}}_+ \dd u \\
    =& \ \mc{J}\bigg[\int_{0}^{\Phi^{-1}(1-3\epsilon)} u\pqty{1 - \frac{1}{\kappa}}\dd u +  \int_{\Phi^{-1}(1-3\epsilon)}^{\kappa\Phi^{-1}(1-3\epsilon)} \pqty{\Phi^{-1}(1-3\epsilon) - \frac{u}{\kappa}} \dd u \bigg ]\\
    =& \ \frac{\mc{J}}{2}(\kappa - 1)[\Phi^{-1}(1-3\epsilon)]^2\\
    \leq &  \ \mc{J}(\kappa - 1)\ln(1/6\epsilon)
\end{align}
Here, $\mc{J} := \frac{1}{a_1 b_2 - a_2 b_1}$ is the Jacobian in the change of variables $(Q_0, P_0) \mapsto (u_1, u_2)$. 
In the above derivation, we have used the fact that $\Phi^{-1}(\Phi(u_1) - 3 \epsilon)$ is upper bounded by both $u$ and $\Phi^{-1}(1-3\epsilon)$ in the first inequality and the Chernoff bound $1-\Phi(x) \leq \frac{1}{2}\ee^{-x^2/2}$ for the quantile in the last inequality.

Let $\lambda^\star = 2/3 + \epsilon^\star$, the optimality condition $|\{S> \lambda^\star\}| = \pi$ and \eqref{eq:excess_set_area_bound} give
\begin{align}\label{eq:lambda_star_bound}
    \frac{2}{3}< \lambda^\star \leq \frac{2}{3} + \frac{1}{6}\ee^{-M}, \ M:= \frac{\pi}{\mc{J}(\kappa - 1)}
\end{align}
so $\lambda^\star$ lies just above $2/3$ by an amount exponentially small in $M$, which scales as $\propto \bqty{\pqty{\bar n +1/2}\gamma_0/\omega_0}^{-1}$ under thermal relaxation \eqref{eq:thermal_wp} and $\bqty{\frac{k_B T}{\hbar \omega_0}\frac{\gamma}{\omega_0}}^{-1}$ under Caldeira-Leggett \eqref{eq:CL_wp}.  

From \eqref{eq:bathtub_lambda}, the gap between $B(2/3)$ and the exact bathtub bound $\mathbf{P}_{\rm wp}$ can be found by the integral 
\begin{align}\label{eq:}
    B(2/3) - B(\lambda^\star) = \int_{\lambda^\star}^{2/3} \frac{\dd B}{\dd \lambda} \dd \lambda =  \int_{2/3}^{\lambda^\star} \pqty{\frac{1}{\pi}|\{s>\lambda\}| - 1} \dd \lambda
\end{align}

Substituting the results in \eqref{eq:excess_set_area_bound} and \eqref{eq:lambda_star_bound} yields the bound 
\begin{align}\label{eq:bathtub_bound_bound}
    B(2/3) - \frac{1}{6}\pqty{1+\frac{1}{M}}\ee^{-M} \leq B(\lambda^\star) \leq B(2/3)
\end{align}
which is, again, exponentially close in $\bqty{\pqty{\pqty{\bar n +\frac{1}{2}}}\frac{\gamma_0}{\omega_0}}^{-1}$ to $B(2/3)$, and admits a closed-form expression. The parametrization of the set $\{S>2/3\}$ in \eqref{eq:excess_set_parametrization} turns the bound at this level to an explicit integral
\begin{align}\label{eq:B(2/3)_def}
    B(2/3) &= \frac{2}{3} + \frac{1}{\pi} \int(S-2/3)_+\dd^2\ZZ_0\\
    &= \frac{2}{3}+\frac{\mc{J}}{\pi} \int_{0}^\infty \int_{-u_1}^{-u_1/\kappa} \frac{1}{3} \bqty{\Phi(u_1)- \Phi(-u_2) }  \dd u_2 \dd u_1 \\
    &= \frac{2}{3} + \frac{\mc{J}(\kappa - 1)^2}{12\pi\kappa}
\end{align}

As $\pqty{\bar n + \frac{1}{2}}\frac{\gamma_0}{\omega_0}$ grows, the lower limit in \eqref{eq:bathtub_bound_bound} slackens and $B(2/3)$ is no longer a tight bound. The threshold is then calculated directly rather than through the surrogate bound $B(2/3)$. We derive here the exact integral for the optimal bathtub bound $\mathbf{P}_{\rm wp} = B(\lambda^\star)$ and detail the numerical search for the exact bound. 

The cross section of $\{S>\lambda \}$ at height $u$ runs from the wedge edge $u/\kappa$ up to the level set $g(u):= \Phi^{-1}(\Phi(u) - 3\epsilon)$, so with $u_\pm$ the two roots of $h(u):=\Phi(u)-\Phi(u/\kappa) = 3\epsilon = 3\lambda - 2$, 

\begin{align}\label{eq:area_quad_app}
  |\{S>\lambda\}| = \mathcal J\!\int_{u_-}^{u_+}\!\bigl[g(u)-u/\kappa\bigr]\,\dd u, \quad
  \int (S-\lambda)_+\,\dd^2\mathbf Z_0 = \frac{\mathcal J}{3}\!\int_{u_-}^{u_+}\!\Bigl[(\Phi(u)-3\epsilon)\bigl(g(u)-u/\kappa\bigr)  - G[g(u)] + G(u/\kappa)\Bigr]\dd u . 
\end{align}
where $G(x):= \int \Phi(x) \dd x = x\Phi(x) +\phi(x)$, and $\phi(x) =\frac{\ee^{-x^2/2}}{\sqrt{2\pi}}$.

The roots $u_\pm$ and the first integral in \eqref{eq:area_quad_app} can be evaluated for a given $\lambda$, and $\lambda^\star$ is found by imposing the area condition $|\{S>\lambda^\star\}| = \pi$. The second integral of \eqref{eq:area_quad_app} evaluated at $\lambda^\star$ then returns the Wigner positivity bound $\mathbf{P}_{\rm wp} = B(\lambda^\star)$. 

\subsection{Thermal relaxation}
For the case of thermal relaxation,
\begin{align*}
    a_k &= \frac{\cos(2\pi k/3)}{\sqrt{\pqty{\bar n +1/2}\pqty{\ee^{\gamma_0 t_k} - 1}}}\\
    b_k &= \frac{\sin(2\pi k/3)}{\sqrt{\pqty{\bar n +1/2}\pqty{\ee^{\gamma_0 t_k} - 1}}}
\end{align*}
and,
\begin{align}\label{eq:thermal_wp}
    \kappa^{\rm th} &= -\frac{b_1}{b_2} =  \sqrt{\frac{\ee^{\gamma_0 t_2} - 1}{\ee^{\gamma_0 t_1} - 1}} \xrightarrow{\gamma_0 \to 0^+} \sqrt{t_2/t_1} = \sqrt{2}\\
    \mc{J}^{\rm th} &= \frac{2}{\sqrt{3}}\pqty{\bar n + 1/2}\sqrt{\pqty{\ee^{\gamma_0 t_1} - 1}\pqty{\ee^{\gamma_0 t_2} - 1}}\xrightarrow{\gamma_0 \to 0^+} \frac{2}{\sqrt{3}}\pqty{\bar n + 1/2}\gamma_0\sqrt{t_1 t_2}\\
    M^{\rm th} &= \frac{\pi}{\mc{J}(\kappa - 1)} \xrightarrow{\gamma_0 \to 0^+} \underbrace{\frac{3\sqrt{3}}{4\sqrt{2}(\sqrt{2} - 1)}}_{2.22}\bqty{\pqty{\bar n + \frac{1}{2}}\frac{\gamma_0}{\omega_0}}^{-1}
\end{align}
Substituting these values to \eqref{eq:B(2/3)_def} and invoking \eqref{eq:bathtub_bound_bound}, we arrive at equation \eqref{eq:photon_bathtub_slope} in the main text for small-$\gamma_0$ scaling of the bathtub bound 
\begin{align}
    \mathbf{P}_{\rm wp}(\gamma_0, \bar n) =& \ \frac{2}{3} + \frac{\bar n + \frac{1}{2}}{6\sqrt{3}\pi}\pqty{\sqrt{\ee^{\gamma_0 t_2} - 1} - \sqrt{\ee^{\gamma_0 t_1} - 1}}^2  - \mc{O}(\ee^{-M^{\rm th}})\\
    =& \ \frac{2}{3} + \frac{3-2\sqrt{2}}{9\sqrt{3}}\pqty{\bar n +\frac{1}{2}}\frac{\gamma_0}{\omega_0} + \mc{O}(\gamma_0^2)
\end{align}

\subsection{Caldeira-Leggett}
For the Caldeira-Leggett model,
\begin{align*}
    a_k&=\frac{\ee^{-\gamma t_k/2}}{\sigma_q^{\,\text{CL}}(t_k)}\Bigl(\cos\Omega t_k+\tfrac{\gamma}{2\Omega}\sin\Omega t_k\Bigr),\\
  b_k&=\frac{\ee^{-\gamma t_k/2}}{m\Omega\,\sigma_q^{\,\text{CL}}(t_k)}\sin\Omega t_k,
\end{align*}

Using the asymptotes in \eqref{eq:sigma_CL_asymptote},
\begin{align}\label{eq:CL_wp}
    \kappa^{\rm CL} &\xrightarrow{\gamma \to 0^+} 1.219\\
    \mc{J}^{\rm CL} &\xrightarrow{\gamma \to 0^+} \underbrace{\frac{2}{\sqrt{3}}\sqrt{\tau_1 \tau_2}}_{3.558}\frac{k_B T}{\hbar \omega_0} 
    \frac{\gamma}{\omega_0}\\
    M^{\rm CL} &\xrightarrow{\gamma \to 0^+} 4.03 \bqty{\frac{k_B T}{\hbar \omega_0} \frac{\gamma}{\omega_0}}^{-1}
\end{align}
Similar to the thermal relaxation case derived above, with these values of $\kappa^{\rm CL}$ and $\mc{J}^{\rm CL}$, we arrive at
\begin{align}
    \mathbf{P}_{\rm wp}^{\rm CL}(\gamma, T) = & \ \frac{2}{3} + \frac{(\sqrt{\tau_2} - \sqrt{\tau_1})^2}{6\sqrt{3} \pi} \frac{k_B T}{\hbar \omega_0}\frac{\gamma}{\omega_0}   - \mc{O}(\ee^{-M^{\rm CL}})\\
    =& \ \frac{2}{3} + 3.71\times 10^{-3}\frac{k_B T}{\hbar \omega_0}\frac{\gamma}{\omega_0}  + \mc{O}(\gamma^2)
\end{align}

\section{Robustness of the protocol under miscalibrated frequency}\label{app:freq_shift}
The protocol's probing times $t_k := \frac{kT_0}{3}$ we considered in this paper are set by a nominal frequency $\omega_0$, which can be off in practice. For example, coupling to the environment might renormalize the frequency at which
the state actually precesses: in the microscopic derivation of \eqref{eq:lindblad_photon}, the principal-value part of the bath integral shifts the oscillator frequency to $\omega = (1+\delta)\omega_0$, which is an analog of the Lamb shift \cite{carmichael1998}. If the frequency is locked to a nominal $\omega_0$ while actually precessing at $\omega$, the protocol probes the dilated phases
\begin{align}
    \theta_k = \omega t_k = \frac{2\pi k}{3}(1+\delta), \quad k\in \{0, 1, 2\}
\end{align}
so the first probe at $k=0$ is exact while the latter are off by $2k\pi \delta/3$. The same dilation describes any miscalibration of frequency, regardless of its origin, whereas it actually precesses at $(1+\delta)\omega_0$, or one can only guarantee that its frequency lies in that band. Either way, the certification threshold must be recomputed for the true, uncertain dynamics. This appendix examines what a bounded frequency uncertainty does to the certification. Formally, a bounded frequency uncertainty replaces the dynamical assumption (A2) of \zcref{sec:protocol} by the weaker

(A2') The dynamics is the one in (A2) with a precession frequency $(1+\delta)\omega_0$ for some unknown $\delta$ with $|\delta| \leq \delta_{\max}$

Under this updated assumption, the valid classical threshold is the worst case over the admissible frequencies
\begin{align}
    \mathbf{P}(\gamma_0, \delta_{\rm max}) := \sup_{|\delta|\leq \delta_{\rm max}} \mathbf{P}(\gamma_0, \delta)
\end{align}
The violation of this bound falsifies any classical model in this range. The Wigner positivity threshold is treated identically. \zcref[S]{fig:mismatch_strip} shows both updated thresholds for the thermal relaxation model at $\bar n = 0$, 
  
\begin{figure}
    \centering  \includegraphics[width=\linewidth]{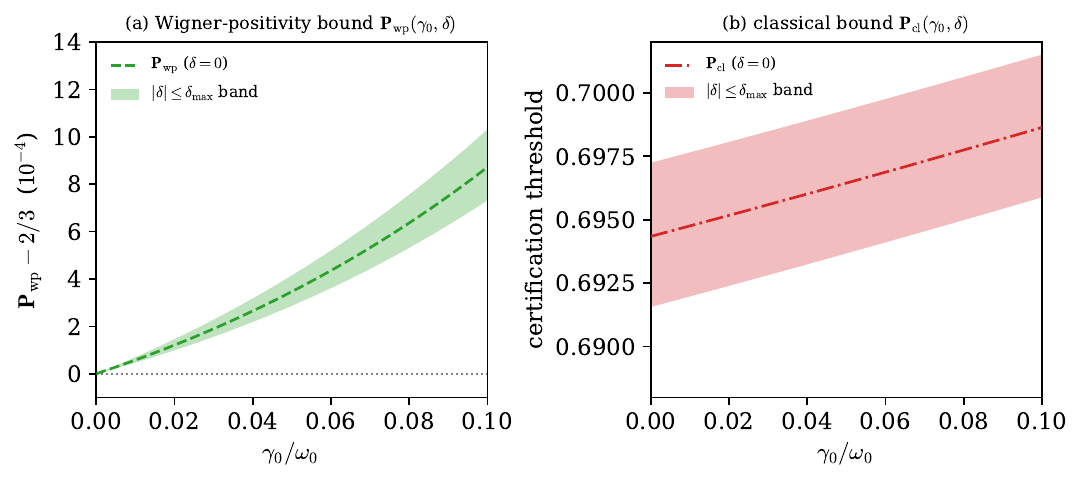}
    \caption{\justifying Certification thresholds under the frequency promise $|\delta| \leq \delta_{\max} = 0.01$ for the thermal relaxation channel at $\bar n = 0$. The centred dashed curve corresponds to the threshold at $\delta = 0$, which is surrounded by the band it sweeps over $|\delta| \leq \delta_{\max}$.}
    \label{fig:mismatch_strip}
\end{figure}
The probing time is a degree of freedom we do not exploit in this paper; we only consider the times $t_k = kT_0/3$ inherited from the noiseless protocol. Optimizing the schedule jointly against the calibrated or worst-case noisy dynamics can only widen the certifiable regions reported here, and we leave it open.

\end{widetext}
\end{document}